\documentclass[compsoc,conference,a4paper,10pt,times]{IEEEtran}
\IEEEoverridecommandlockouts

\usepackage{times}
\usepackage[colorlinks=true,urlcolor=black]{hyperref}
\usepackage{url}
\usepackage{array,multirow}
\usepackage{xspace}
\usepackage{xcolor}
\usepackage{color, colortbl}
\usepackage{balance}
\usepackage{paralist}
\usepackage{cleveref}
\usepackage{amsfonts,amssymb}
\usepackage{amsthm,amsmath}
\usepackage{graphicx}
\usepackage{pifont}
\usepackage{array}
\usepackage{booktabs}
\usepackage{tikz}
\usepackage{bm}
\usepackage{enumitem}
\usepackage{todonotes}
\usepackage{arydshln}
\usepackage{makecell}
\usepackage{tabularx}
\usepackage[usestackEOL]{stackengine}
\usepackage{authblk}
\usepackage{caption}
\usepackage{textcomp}
\usepackage{bmpsize}
\usepackage{xcolor}
\usepackage{cite}

\def\BibTeX{{\rm B\kern-.05em{\sc i\kern-.025em b}\kern-.08em
    T\kern-.1667em\lower.7ex\hbox{E}\kern-.125emX}}

\newcommand\sysname{Byzcuit\xspace}
\newcommand\sysnamereplay{\texttt{byzcuit}\xspace}
\newcommand\sysnamebaseline{\texttt{byzcuit-baseline}\xspace}

\newcommand\chainspace{Chainspace\xspace}
\newcommand\ethereum{Ethereum\xspace}

\newcommand\omniledger{Omniledger\xspace}
\newcommand\rapidchain{RapidChain\xspace}

\newcommand\sbac{S-BAC\xspace}

\newcommand\atomix{Atomix\xspace}

\newcommand\rscoin{RSCoin\xspace}

\newcommand\bftsmart{\textsc{bft-SMaRt}\xspace}

\newcommand\shard[1]{\emph{shard}\xspace#1\xspace}
\newcommand\Shard[1]{\emph{Shard}\xspace#1\xspace}
\newcommand\preacceptt{\textsf{pre-accept}($T$)\xspace}
\newcommand\preabortt{\textsf{pre-abort}($T$)\xspace}
\newcommand\preaccepttt{\textsf{pre-accept}($T'$)\xspace}

\newcommand\preacceptts{\textsf{pre-accept}($T,s_T$)\xspace}
\newcommand\preabortts{\textsf{pre-abort}($T,s_T$)\xspace}

\newcommand\acceptt{\textsf{accept}($T$)\xspace}
\newcommand\abortt{\textsf{abort}($T$)\xspace}

\newcommand\aborttt{\textsf{abort}($T'$)\xspace}

\newcommand\acceptts{\textsf{accept}($T,s_T$)\xspace}
\newcommand\abortts{\textsf{abort}($T,s_T$)\xspace}
\newcommand\myrow[1]{row\xspace\textsf{#1}\xspace}
\newcommand\mycolumn[1]{column\xspace\textsf{#1}\xspace}
\newcommand\shardled{shard-led\xspace}
\newcommand\Shardled{Shard-led\xspace}
\newcommand\clientled{client-led\xspace}
\newcommand\Clientled{Client-led\xspace}
\newcommand\activeObj{`active'\xspace}
\newcommand\inactiveObj{`inactive'\xspace}
\newcommand\locked{`locked'\xspace}

\newcommand\attacker{attacker\xspace}

\newcommand\prerecorded{prerecorded\xspace}
\newcommand\prerecords{prerecords\xspace}

\newcommand{\eg}{\textit{e.g.}\@\xspace}
\newcommand{\ie}{\textit{i.e.}\@\xspace}
\newcommand{\via}{\textit{via}\@\xspace}

\newcommand\mypara[1]{\vspace{0.05in} \noindent \textbf{#1.}}
\newcommand\para[1]{\vspace{0.05in} \noindent \textbf{#1.}}

\def\first{({\it i})\xspace}
\def\second{({\it ii})\xspace}
\def\third{({\it iii})\xspace}
\def\fourth{({\it iv})\xspace}

\definecolor{verylightgray}{gray}{0.8}

\newcolumntype{L}{l<{\hspace{1cm}}}
\newcolumntype{C}{c<{\hspace{1cm}}}
\newcolumntype{D}{c<{\hspace{0.3cm}}}

\newtheorem{theorem}{Theorem}

\newtheorem{lemma}{Lemma}
\newtheorem{SecAssumption}{Assumption}

\begin{document}

\title{
Replay Attacks and Defenses Against Cross-shard Consensus in Sharded Distributed Ledgers
}

\author{
Alberto Sonnino, Shehar Bano, Mustafa Al-Bassam, George Danezis \\
University College London \& \texttt{chainspace.io}
}

\maketitle

\begin{abstract}
We present a family of replay attacks against sharded distributed ledgers targeting cross-shard consensus protocols, such as the recently proposed \chainspace and \omniledger. They allow an attacker, with network access only, to double-spend or lock resources with minimal efforts. The attacker can act independently without colluding with any nodes, and succeed even if all nodes are honest; most of the attacks can also exhibit themselves as faults under periods of asynchrony. These attacks are effective against both \shardled and \clientled cross-shard consensus approaches. 
We present \sysname---a new cross-shard consensus protocol that is immune to those attacks. We implement a prototype of \sysname and evaluate it on a real cloud-based testbed, showing that our defenses impact performance minimally, and overall performance surpasses previous works.
\end{abstract}


\section{Introduction} 

\label{sec:introduction}

Sharding is one of the key approaches to address blockchain scalability issues~\cite{Bano:2017b}, and a growing number of systems are implementing sharded blockchains~\cite{chainspace, kokoris2017omniledger, rscoin, rapidchain, elastico, Bano:2017b}.  
 The key idea is to create groups (or shards) of nodes that handle only a subset of all transactions and system state, relying on classical Byzantine Fault Tolerance (BFT) protocols for reaching \emph{intra-shard consensus}. These systems achieve optimal performance and scalability because: \first non-conflicting transactions can be processed in parallel by multiple shards; and \second the system can scale up by adding new shards. 
This separation of transaction handling across shards is not perfectly `clean'---a transaction might rely on data managed by multiple shards, requiring an additional step of \emph{cross-shard consensus} across the concerned shards.
An atomic commit protocol (such as the two-phase commit protocol~\cite{gray1978notes}) typically runs across all the concerned shards 
to ensure the transaction is accepted by all or none of those shards.



We present the first replay attacks on cross-shard consensus in sharded blockchains. %
An \attacker can launch these attacks with minimal effort, without subverting any nodes, and assuming a weakly synchronous network (and in some cases, without relying on any network assumption)---even when the byzantine safety assumptions are satisfied.
These attacks compromise key system properties of safety and liveness, effectively enabling the \attacker to double-spend coins (or any other objects managed by the blockchain) and create coins out of thin air.
Our attacks apply to the two main approaches to achieve cross-shard consensus~\cite{Bano:2017b}: \first \shardled protocols that only involve the concerned shards, and require no external entity for coordination (\Cref{shard-led-consensus}); and \second \clientled protocols that are coordinated by the client (\Cref{client-led-consensus}).

We concretely sketch the replay attacks in the context of two representative systems: \chainspace~\cite{chainspace}
as an example of \shardled protocols; and \omniledger~\cite{kokoris2017omniledger}  
as an example of \clientled protocols. Not only those systems were recently presented at top security conferences, but they form the basis of numerous start-ups and open-source projects such as chainspace.io\footnote{https://chainspace.io} and Harmony\footnote{https://harmony.one}. 
For each of the two cross-shard consensus approaches, Appendix~\ref{sec:prerecoring} describes how an \attacker can actively stage the attack by eliciting from the system the messages to replay (in contrast to passively observing the network traffic, and waiting to detect and record the target messages).
We also discuss the feasibility of these attacks and their real-world impact; and we responsibly disclosed them to the concerned companies.
We implement and open-source\footnote{\url{https://github.com/sheharbano/byzcuit/tree/replay-attacks}} a demo of the replay attacks described in this paper (Section~\ref{attack-overview}).

The replay attacks we present are generic and apply to other systems that are based on similar models, like \rapidchain~\cite{rapidchain}. 
\ethereum's cross-shard ``yanking'' proposal~\cite{yanking} also faces similar challenges; \Cref{sec:ethereum} describes their cross-shard consensus protocol and compares their current proposal to mitigate cross-shard replay attacks with this work. We note that account-based blockchains like \ethereum defend against transactions replay using account sequence numbers, in an entirely different context; \ie each account holds a monotonically increasing counter to prevent attackers from re-submitting old transactions to the network.
On the other hand, this work focuses on attacks due to replaying messages in cross-shard atomic commit protocols. 
Based on our detailed analysis of replay attacks, we develop a defense strategy (\Cref{defense}).

Drawing insights from our analysis of performance trade-offs and replay attack vulnerabilities in existing \shardled and \clientled cross-shard consensus protocols, we present a hybrid system, \sysname (\Cref{sec:byzcuit}). It combines useful features from both these design approaches to achieve high performance and scalability, and leverages our proposed defense to achieve resilience against replay attacks. \sysname employs a Transaction Manager to coordinate cross-shard communication, reducing its cost to $O(n)$ communication, between $n$ shards, in the absence of faults.
We implement a prototype of \sysname in Java as a fork of the \chainspace code~\cite{chainspace}, and release it as an open-source project\footnote{\url{https://github.com/sheharbano/byzcuit}}. We evaluate \sysname on a real cloud-based testbed under varying transaction loads and show that \sysname has a client-perceived latency of less than a second, even for a system load of 1000 transactions per second (tps). \sysname's transaction throughput scales linearly with the number of shards by 250--300 tps for each shard added, handling up to 1550 tps with 10 shards---which is about 8 times higher than \chainspace with a similar setup. We quantify the overhead of our replay defenses and find that as expected those reduce the throughput by 20--250 tps.

\para{Contributions} We make the following key contributions: we \first develop the first replay attacks against \shardled and \clientled cross-shard consensus protocols, and illustrate their impact on important academic and implemented designs; and 
\second present defenses; 
\third design a hybrid, new system \sysname with improved performance trade-offs, and which integrates our proposed defense to achieve resilience against the replay attacks; 
and \fourth, we implement a prototype of \sysname and evaluate its performance and scalability on a real distributed set of nodes and under varying transaction loads, and illustrate how it is superior to previous proposals.





\section{Background and Related Work}
\label{background}
We present background and related work on cross-shard consensus protocols.


\mypara{Sharded blockchains}
%
 Earlier systems like Bitcoin~\cite{bitcoin} probabilistically elect a single node which can extend the blockchain. However, such systems assume synchrony, have no finality (\ie, forks can exist and be eventually accepted) and low performance (\ie, high latency and low throughput). Consequently, there has been a shift to committee-based designs~\cite{Bano:2017b} where a group of nodes collectively extends the blockchain typically \via classical byzantine fault tolerance (BFT) consensus protocols such as PBFT~\cite{pbft}. While these systems offer better performance, single-committee consensus is not scalable---as every node handles every transaction, adding more nodes to the committee decreases throughput due to the increased communication overhead. 
This motivated the design of \emph{sharded} systems, where multiple committees handle a subset of all the transactions---allowing parallel execution of transactions. Every committee has its own blockchain and set of objects (or unspent transaction outputs, UTXO) that they manage. Committees run an \emph{intra-shard consensus protocol} (\eg, PBFT) within themselves, and extend the blockchain in parallel. 

\mypara{Cross-shard consensus}
In sharded systems, some transactions may operate on objects handled by different shards, effectively requiring the relevant shards to additionally run a \emph{cross-shard consensus protocol} to enable agreement across the shards. If any of the shards relevant to the transaction rejects it, all the other shards should likewise reject the transaction to ensure atomicity.

The typical choice for implementing cross-shard consensus is the two-phase atomic commit protocol~\cite{gray1978notes}. This protocol has two phases which are run by a \emph{coordinator}. In the first \emph{voting} phase, the nodes tentatively write changes locally, lock resources, and report their status to the coordinator. If the coordinator does not receive status message from a node (\eg, because the node crashed or the status message was lost), it assumes that the node's local write failed and sends a rollback message to all the nodes to ensure any local changes are reversed, and locks released. If the coordinator receives status messages from all the nodes, it initiates the second \emph{commit} phase and sends a commit message to all the nodes so they can permanently write the changes and unlock resources. In the context of sharded blockchains, the atomic commit protocol operates on shards (which make the local changes associated with the voting phase \via an intra-shard consensus protocol like PBFT), rather than individual nodes. A further consideration is who assumes the role of the coordinator.     

\mypara{Related Work}
Replay attacks in general have seen extensive study in the security literature.
This is the first paper that presents replay attacks on cross-shard consensus protocols.
Traditionally, the most stringent threat models assumed by consensus protocols involve byzantine adversaries who are able to control or subvert consensus nodes and cause them to behave arbitrarily. 
Repurposing those protocols to open permissionless networks (\eg, blockchains) opens up new attack avenues such as replay attacks as shown in this paper.
There are currently two key approaches to cross-shard consensus~\cite{Bano:2017b}. The first approach involves \emph{\clientled protocols} (such as \omniledger~\cite{kokoris2017omniledger} and \rscoin~\cite{rscoin}), where the client acts as a coordinator. These protocols assume that clients are incentivized to proceed to the unlock phase. While such incentives may exist in a cryptocurrency application where an unresponsive client loses its own coins if the inputs are permanently locked, these do not however hold for a general-purpose platform where transaction inputs may have shared ownership.
The second approach involves \emph{\shardled protocols} (such as \chainspace~\cite{chainspace}, Rapidchain~\cite{rapidchain} and Elastico~\cite{elastico}), where shards collectively assume the role of a coordinator.
All the concerned shards collaboratively run the protocol between them. This is achieved by making the entire shard act as a `resource manager' for the transactions it handles.
We describe our replay attacks in the context of two representative systems: \chainspace~\cite{chainspace}  
as an example of \shardled protocols (\Cref{shard-led-consensus}); and \omniledger~\cite{kokoris2017omniledger}  
as an example of \clientled protocols (\Cref{client-led-consensus}).
We provide a more detailed description of these systems in the relevant sections.

\section{Attack Overview} 

\label{attack-overview}

Sections~\ref{shard-led-consensus}~and~\ref{client-led-consensus} discuss replay attacks on both \shardled and \clientled cross-shard consensus protocols, respectively. We present a high-level description of these attacks, the threat model, demo attack implementation, and the notation used in this paper.

\mypara{Replay Attacks on Cross-Shard Consensus}
The \attacker records a target shard's responses to 
the atomic commit protocol, and replays them during another instance of the protocol. We present 
\first attacks against the first phase (\emph{voting}), and \second attacks against the second phase (\emph{commit}) of the atomic commit protocol.

To attack the first phase (\emph{voting}) of the atomic commit protocol, the \attacker replaces messages generated by the target shard by replaying pre-recorded messages. In practice, the \attacker does not \emph{replace} those messages---it achieves a similar result by making its replayed messages arrive at the coordinator faster (racing the target shard's original message), exploiting the fact that the coordinator makes progress based on the first message it receives. Replaying messages in this fashion enables the \attacker to compromise the system safety (by creating inconsistent state on the shards) and/or liveness (by causing valid transactions to be rejected). 

To attack the second phase (\emph{commit}) of the atomic commit protocol, the \attacker simply replays prerecorded messages to target shards, and compromises consistency. The \attacker can replay those messages at any time of its choice, and does not rely on any racing condition as in the previous case.


\mypara{Threat Model} 
The \attacker can successfully launch the described attacks without colluding with any shard nodes, and under the BFT honest majority safety assumption for nodes within shards (\ie, the attacks are effective even if \emph{all} nodes are honest). We assume an \attacker that can observe and record messages generated by shards; this can be achieved by \first monitoring the network, or \second reading the blockchain (which is more practical).
The \attacker can be an external observer that passively collects the target messages at the level of the network, or it can act as a client and actively interact with the system to elicit the target messages.
The attacks against the first phase of the atomic commit protocol (Sections~\ref{chainspace-attack-1} and \ref{omniledger-attack-1}) assume a weakly synchronous network  in which an \attacker may delay messages and race target shards by replaying pre-recorded messages. 
The attacks against the second phase of the atomic commit protocol (Section~\ref{chainspace-attack-2}~and~\ref{omniledger-attack-2}) do not make any such assumptions on the underlying network. 

\para{Attack Implementation}
We implemented a demo of the replay attacks against \chainspace~\cite{chainspace}, as an example of systems that implement shard-led cross-shard consensus protocol, in Java.\footnote{Attacks against systems with client-led cross-shard consensus such as \omniledger~\cite{kokoris2017omniledger} can be similarly implemented.} 
We are open-sourcing the demo of our attacks\footnote{\url{https://github.com/sheharbano/byzcuit/tree/replay-attacks}}, and a document describing a step-by-step tutorial to execute the attacks\footnote{\url{https://github.com/sheharbano/byzcuit/blob/master/docs/Chainspace-Replay-Attack-Demo.pdf}}. 
The demo shows, in the context of a simple payment application that supports account creation and coin transfer, how an attacker can use the replay attacks described in this paper to create coins out of thin air. 
Note that the attacks do not rely on any strict timing assumptions---the same attacker could control the accounts of both payer and payee, as well as the client.

\mypara{Notation}
Operations on the blockchain are specified as \emph{transactions}. A transaction defines some transformation on the blockchain state, and has input and output \emph{objects} (such as UTXO entries). An object is some data managed by the blockchain, such as a bank account, a specific coin, or a hotel room. 
 For example, $T(x_1,x_2)\rightarrow (y_1,y_2,y_3)$ represents a transaction with two inputs, $x_1$ managed by \shard{1} and $x_2$ managed by \shard{2}; and three outputs, $y_1$ managed by \shard{1}, $y_2$ managed by \shard{2}, and $y_3$ managed by \shard{3}.
We call the shards that manage the input objects \emph{input shards}, and the shards that manage the output objects \emph{output shards}. It is possible for a shard to be both the input and output shard.
Objects can be in two states: \emph{active} (on unspent) objects are available for being processed by a transaction; and \emph{inactive} (or spent) objects cannot be processed by any transaction.
Additionally, some systems also associate \emph{locked} state with objects that are currently being processed by a transaction to protect against manipulation by other concurrent transactions involving those objects. 
The attacks we describe in this paper generalize to transactions with $k$ inputs and $k'$ outputs managed by an arbitrary number of shards.


\section{\Shardled Cross-Shard Consensus} 
\label{shard-led-consensus}

\begin{figure}[t]
\centering
\includegraphics[width=.45\textwidth]{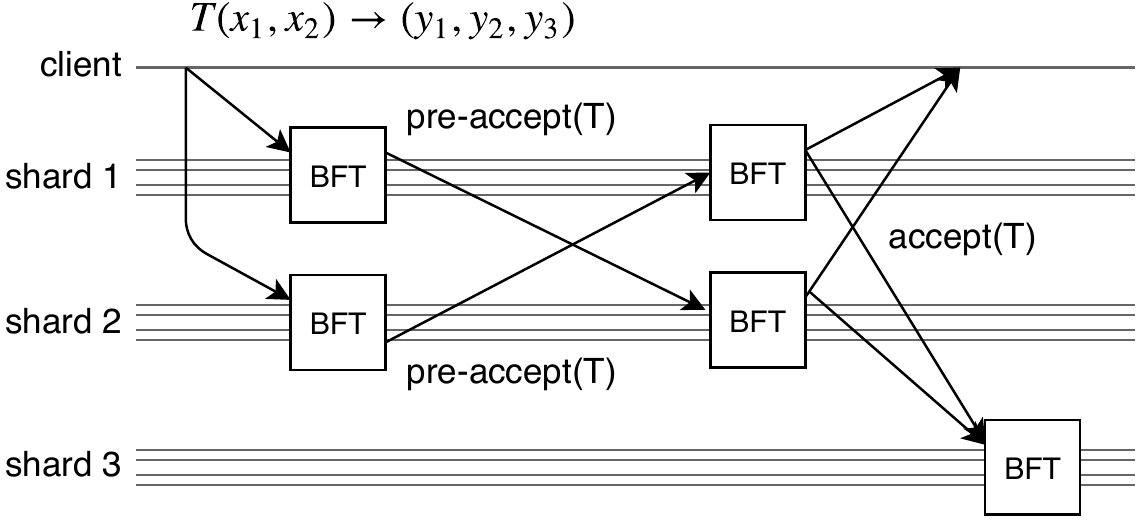}
\caption{\footnotesize An example execution of \sbac for a valid transaction $T(x_1,x_2)\rightarrow (y_1,y_2,y_3)$ with two inputs ($x_1$ and $x_2$, both are active) and three outputs ($y_1,y_2,y_3$), where the final decision is \acceptt. 
}
\label{fig:sbac}
\end{figure}

In \shardled cross-shard consensus protocols, the shards collectively take on the role of the coordinator in the atomic commit protocol. We describe replay attacks on \shardled cross-shard consensus protocols.  To make the discussion concrete, we illustrate these attacks in the context of \chainspace~\cite{chainspace} (\Cref{chainspace}), though we note that these attacks can be generalized to other similar systems. We discuss how the \attacker can record shard messages to replay in future attacks (\Cref{chainspace-message-recording}). In Sections~\ref{chainspace-attack-1}~and~\ref{chainspace-attack-2}, we describe replay attacks on the first and second phase of the cross-shard consensus protocol, and discuss the real-world impact of these attacks (\Cref{chainspace-discussion}).  

\subsection{\chainspace Overview} 
\label{chainspace}

\chainspace uses a \shardled cross-shard consensus protocol called \sbac. 
The client submits a transaction to the input shards.
Each shard internally runs a BFT protocol to tentatively decide whether to accept or abort the transaction locally, and broadcasts its local decision (\preacceptt or \preabortt) to other relevant shards.
\Cref{fig:sbac-original-fsm} shows the state machine representing the life cycle of objects in \chainspace. A shard generates \preabortt if the transaction fails local checks (\eg, if any of the input objects are \inactiveObj or \locked).
 If a shard generates \preacceptt, it changes the state of the input objects to \locked.
This is the first step of \sbac, and is equivalent to the voting phase in the two-phase atomic commit protocol (\Cref{background}).

Each shard collects responses from other relevant shards, and commits the transaction if all shards respond with \preacceptt, or aborts the transaction otherwise.
This is the second step of \sbac, and is equivalent to the commit phase in the two-phase atomic commit protocol (\Cref{background}).
The shards communicate this decision to the client as well as the output shards by sending them the \acceptt or \abortt messages.
If the shard's decision is \acceptt, it changes the input object state to \inactiveObj.
If the shard's decision is \abortt, it changes the input object state to \activeObj (effectively unlocking it).
Upon receiving \acceptt, the client concludes that the transaction was committed, and the output shards create the output objects (with the state \activeObj) of the transaction.

\begin{figure}[t]
\centering
\includegraphics[width=.45\textwidth]{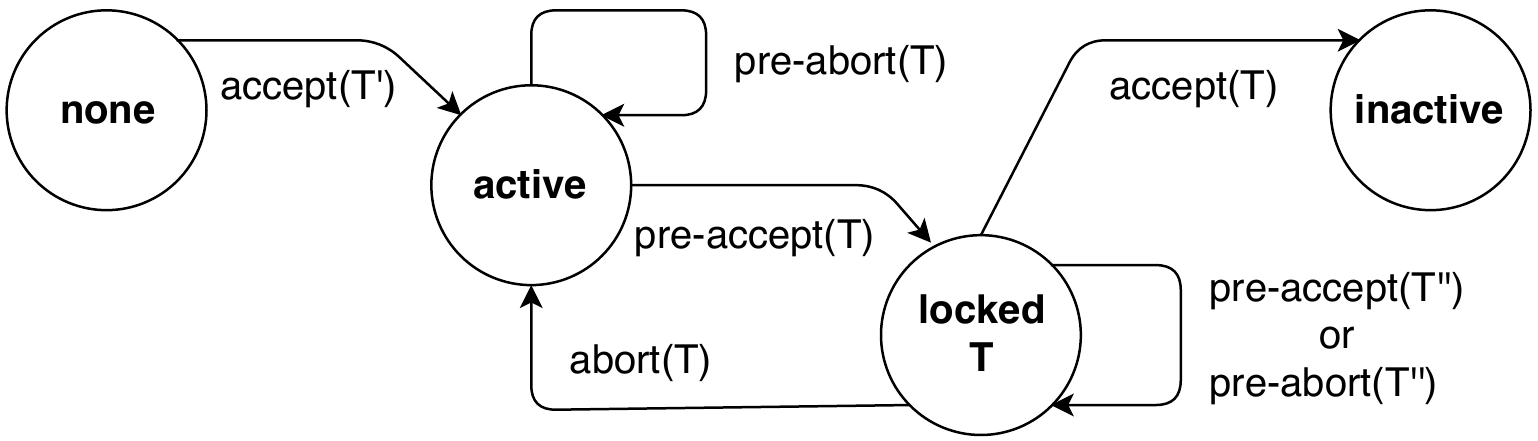}
\caption{\footnotesize State machine representing the life cycle of \chainspace objects. An object becomes \activeObj as a result of a previous successful transaction. The object state changes to \locked if a shard locally emits \preacceptt in the first phase of the cross-shard consensus protocol for a transaction $T$. A locked object cannot be processed by other transactions $T{''}$. If the second phase of the protocol results in \acceptt, the object becomes \inactiveObj; alternatively, if the result is \abortt the object becomes \activeObj again and is available for being processed by other transactions.}
\label{fig:sbac-original-fsm}
\end{figure}

\Cref{fig:sbac} shows an example execution of \sbac for a valid transaction $T(x_1,x_2)\rightarrow (y_1,y_2,y_3)$ with two inputs ($x_1$ and $x_2$, both are active) and three outputs ($y_1,y_2,y_3$), where the final decision is \acceptt. 
The client submits $T$ to \shard{1} and \shard{2}. 
Upon receiving $T$, both \shard{1} and \shard{2} confirm that the transaction is to commit, and emit \preacceptt at the end of the first phase of \sbac. 
Each shard receives \preacceptt from the other shard, and emits \acceptt at the end of the second phase of \sbac. 
As a result, the input objects $x_1$ and $x_2$ become inactive, and the output shards respectively create objects $y_1$, $y_2$, and $y_3$.

\subsection{Message Recording}

\label{chainspace-message-recording}

Prior to the replay attacks, the \attacker records responses generated by shards. The \attacker can record shard responses in the first phase of \sbac (\ie, \preacceptt or \preabortt), enabling the family of attacks described in \Cref{chainspace-attack-1}. The \attacker can also record shard responses in the second phase of \sbac (\ie, \acceptt or \abortt), enabling the family of attacks described in \Cref{chainspace-attack-2}. 
In the general case, the attacker passively collects the messages either by sniffing the network on protocol executions, or by downloading the blockchain and selecting the messages to replay\footnote{Since those messages need to be recorded on chain for verification, just using transport layer encryption between nodes is not effective.}. \Cref{sec:prerecoring-chainspace} shows how the \attacker can act as client to actively elicit the messages necessary for the attacks---this empowers the \attacker to actively orchestrate the attacks.


\subsection{Attacks on the First Phase of \sbac} 
\label{chainspace-attack-1}

\begin{table*}[t]
\centering
\footnotesize
\begin{tabular}{c c C  c c c}
\toprule
\multicolumn{3}{c}{\small Phase 1 of \sbac} & \multicolumn{3}{c}{\small Phase 2 of \sbac}\\
\noalign{\smallskip}
&\Centerstack{\textbf{Shard 1} \\ (potential victim)} & \Centerstack{\textbf{Shard 2} \\ (potential victim)} & \Centerstack{\textbf{Shard 1} \\ (potential victim)} & \Centerstack{\textbf{Shard 2} \\ (potential victim)} & \Centerstack{\textbf{Shard 3} \\ (potential victim)}\\
\noalign{\smallskip}
\toprule

\rowcolor{verylightgray}
1 & \Centerstack{\preacceptt \\ lock $x_1$} & \Centerstack{\preacceptt \\ lock $x_2$} &
\Centerstack{\acceptt \\ create $y_1$; inactivate $x_1$} & \Centerstack{\acceptt \\ create $y_2$; inactivate $x_2$} & \Centerstack{- \\ create $y_3$} \\
\midrule
2 & $\rhd$\preabortt & & \Centerstack{\acceptt \\ create $y_1$; inactivate $x_1$} & \Centerstack{\abortt \\ unlock $x_2$} & \Centerstack{- \\ create $y_3$} \\
\midrule
3 & & $\rhd$\preabortt & \Centerstack{\abortt \\ unlock $x_1$} & \Centerstack{\acceptt \\ create $y_2$; inactivate $x_2$} & \Centerstack{- \\ create $y_3$} \\
\midrule
4 & $\rhd$\preabortt & $\rhd$\preabortt & \Centerstack{\abortt \\ unlock $x_1$} & \Centerstack{\abortt \\ unlock $x_2$} & - \\
\noalign{\smallskip}
\toprule

\rowcolor{verylightgray}
5 & \Centerstack{\preabortt \\ -} & \Centerstack{\preacceptt \\ lock $x_2$} & \Centerstack{\abortt \\ -} & \Centerstack{\abortt \\ unlock $x_2$} & - \\
\midrule
6 & $\rhd$\preacceptt & & \Centerstack{\abortt \\ -} & \Centerstack{\acceptt \\ create $y_2$; inactivate $x_2$} &  \Centerstack{- \\ create $y_3$} \\
\noalign{\smallskip}
\toprule

\rowcolor{verylightgray}
7 & \Centerstack{\preacceptt \\ lock $x_1$} & \Centerstack{\preabortt \\ -} & \Centerstack{\abortt \\ unlock $x_1$} & \Centerstack{\abortt \\ -} & - \\
\midrule
8 & & $\rhd$ \preacceptt& \Centerstack{\acceptt \\ create $y_1$; inactivate $x_1$} & \Centerstack{\abortt \\ -} & \Centerstack{- \\ create $y_3$} \\
\noalign{\smallskip}
\toprule

\rowcolor{verylightgray}
9 & \Centerstack{\preabortt \\ -} & \Centerstack{\preabortt \\ -} & \Centerstack{\abortt \\ -} & \Centerstack{\abortt \\ -} & - \\
\bottomrule
\end{tabular}
\caption{\footnotesize List of replay attacks against the first phase of \sbac for all possible executions of the transaction $T(x_1,x_2)\rightarrow (y_1,y_2,y_3)$ as described in \Cref{attack-overview}. 
The highlighted rows indicate correct executions of \sbac (\ie, without the \attacker), and the other rows indicate incorrect executions due to the replay attacks.
 In multirows, the top sub-rows show the protocol messages emitted by shards, and the bottom sub-rows indicate local shard actions as a result of emitting those messages.
 For example, (\mycolumn{3}, \myrow{2}) means that \shard{1} emits \acceptt (top sub-row), and creates a new object $y_1$ and inactivates $x_1$ (bottom sub-row).  
The first two columns indicate the messages emitted by each shard at the end the first phase of \sbac.
The \attacker races shards at the end of the first phase of \sbac by replaying \prerecorded messages, marked with the symbol $\rhd$ in the first two columns of \Cref{tab:sbac-attack}.
For example $\rhd$\preabortt at (\mycolumn{1}, \myrow{2}) means that the \attacker sends to other relevant shards (in this case \shard{2}) a \prerecorded \preabortt\xspace message impersonating \shard{1} that races the original \preacceptt (\mycolumn{1}, \myrow{1}) emitted by \shard{1}. 
The last three columns indicate the messages emitted at the end of the second phase of \sbac.}
\label{tab:sbac-attack}
\end{table*}

\begin{figure}[t]
\centering
\includegraphics[width=.45\textwidth]{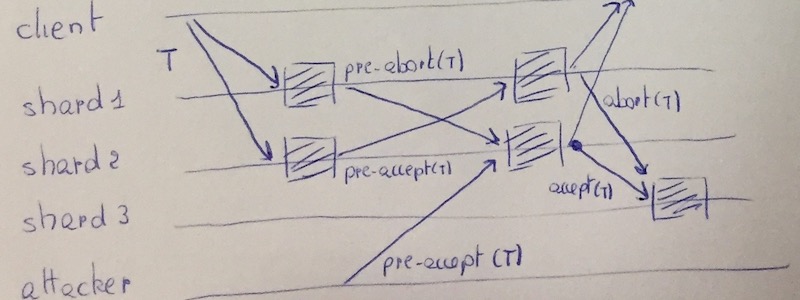}
\caption{\footnotesize 
Illustration of the replay attack depicted in row \textsf{6} of \Cref{tab:sbac-attack}. The \attacker replays to \shard{2} a \prerecorded \preacceptt message (shown as a bold line) from \shard{1}, which precludes \shard{1}'s \preabortt message (shown as a dotted line). 
}
\label{fig:sbac-attack-example-2}
\end{figure}

We present replay attacks on the first phase of \sbac by taking the example of a transaction $T(x_1,x_2)\rightarrow (y_1,y_2,y_3)$ as described in \Cref{attack-overview}. These attacks easily generalize to transactions with $k$ inputs and $k'$ outputs managed by an arbitrary number of shards. The replay attacks work in two steps; \first the \attacker records \preacceptt or \preabortt messages (as described in \Cref{chainspace-message-recording} and \Cref{sec:prerecoring-chainspace}); and \second then replays those messages during a new instance of the protocol.

\Cref{tab:sbac-attack} shows the replay attacks that the \attacker can launch, for all possible combinations of messages emitted by \shard{1} and \shard{2} in the first phase of \sbac. The caption includes details about how to interpret this table. All attacks exploit the parallel composition of multiple \sbac instances, and insufficient binding of messages to its \sbac instance.
We describe \myrow{6} of \Cref{tab:sbac-attack}, to help readers interpret rest of the table on their own.
In the correct execution (\myrow{5}), \shard{1} and \shard{2} emit \preabortt (because $x_1$ is not active) and \preacceptt in the first phase, respectively. In the second phase, both shards emit \abortt and the protocol terminates.   
\Cref{fig:sbac-attack-example-2} illustrates the replay attack corresponding to row \textsf{6} of \Cref{tab:sbac-attack}.
The attacker races \shard{1} by sending to \shard{2} the \prerecorded \preacceptt message from \shard{1}. As a result, \shard{2} emits \acceptt, inactivates object $x_2$ and creates object $y_2$. This leads to inconsistent state across the shards. In a correct execution: \first if $T$ is accepted all its inputs ($x_1$ and $x_2$) should become inactive, and all the outputs ($y_1$, $y_2$, $y_3$) should be created; and \second if $T$ is aborted, all its inputs ($x_1$ and $x_2$) should become active again, and none of the outputs ($y_1$, $y_2$, $y_3$) should be created. However, here we have an incorrect termination of \sbac: at the end of the protocol $x_1$ could be active and $x_2$ is inactive; $y_1$ is not created, $y_2$ and $y_3$ are created.

\Cref{tab:sbac-attack} shows that through careful selection of the messages to replay from different \sbac instances, the attacks can be effective against any shard.
All the attacks (except \myrow{4}) compromise consistency; the \attacker can trick the input shards to inactivate arbitrary objects, and trick the output shards into creating new objects in violation of the protocol. 
The attack depicted in \myrow{4} only affects availability.

\subsection{Attacks on the Second Phase of \sbac} 
\label{chainspace-attack-2}

\begin{table*}[t]
\centering
\footnotesize
\begin{tabular}{c C C C}
\toprule
\multicolumn{4}{c}{\small Phase 2 of \sbac}\\
& \bf Shard 1 & \bf Shard 2 & \shortstack{\textbf{Shard 3} \\ (potential victim)}\\
\midrule

\rowcolor{verylightgray}
1 & \Centerstack{\acceptt \\ create $y_1$; inactivate $x_1$} & \Centerstack{\acceptt \\ create $y_2$; inactivate $x_2$} & \Centerstack{- \\ create $y_3$} \\
2 & $\rhd$\acceptt & & create $y_3$ \\
3 & & $\rhd$\acceptt & create $y_3$ \\
4 & $\rhd$\acceptt & $\rhd$\acceptt & create $y_3$ \\
\midrule
\rowcolor{verylightgray}
5 & \Centerstack{\abortt \\ (unlock $x_1$)} & \Centerstack{\abortt \\ (unlock $x_2$)} & \Centerstack{- \\ -} \\
6 & $\rhd$\acceptt & & create $y_3$ \\
7 & & $\rhd$\acceptt & create $y_3$ \\
8 & $\rhd$\acceptt & $\rhd$\acceptt & create $y_3$ \\

\bottomrule
\end{tabular}
\caption{\footnotesize  List of replay attacks against the second phase of \sbac for all possible executions of the transaction $T(x_1,x_2)\rightarrow (y_1,y_2,y_3)$ as described in \Cref{attack-overview}. 
The highlighted rows indicate correct executions of \sbac (\ie, without the \attacker), and the other rows indicate incorrect executions due to the replay attacks.
In multirows, the top sub-rows show the protocol messages emitted by shards, and the bottom sub-rows indicate local shard actions as a result of emitting those messages.
 For example, (\mycolumn{1}, \myrow{1}) means that \shard{1} emits \acceptt (top sub-row), and creates a new object $y_1$ and inactivates $x_1$ (bottom sub-row).
The first two columns indicate the messages emitted by each shard at the end the second phase of \sbac, and the last column shows the effect of these messages on the output \shard{3}.
Replayed messages are marked with the symbol $\rhd$.
For example $\rhd$\acceptt at (\mycolumn{1}, \myrow{2}) means that the \attacker sends to other relevant shards (in this case \shard{3}) a \prerecorded \acceptt message impersonating \shard{1}.} 
\label{tab:sbac-attack:2}
\end{table*}

We present replay attacks on the second phase of \sbac. The \attacker \prerecords \acceptt messages as described in \Cref{chainspace-message-recording} and \Cref{sec:prerecoring-chainspace}. 
\Cref{tab:sbac-attack:2} shows replay attacks for all possible combinations of messages emitted by \shard{1} and \shard{2} in the second phase. Since the attacks we describe in this section assume that the first phase of \sbac concluded correctly (\ie, all the relevant shards unanimously decide to accept or reject a transaction), both the shards generate \abortt (\myrow{1}) or \acceptt (\myrow{5}).
The caption includes details about how to interpret this table.
We describe \myrow{6} of \Cref{tab:sbac-attack:2}, to help readers interpret rest of the table on their own.
In the correct execution (\myrow{5}), both shards emit \abortt and no output objects are created.
 In the attack in \myrow{6}, the \attacker replays a \prerecorded \acceptt from \shard{1} to all the relevant shards (in this case \shard{3}).
Upon receiving this message, \shard{3} (incorrectly) creates $y_3$.

The potential victims of replay attacks corresponding to the second phase of \sbac are the shards that \emph{only} act as output shards (\ie, do not simultaneously act as input shards). 
The \attacker can replay \acceptt multiple times tricking \shard{3} into creating $y_3$ multiple times. 
These attacks are possible because shards do not keep records of inactive objects (following the UTXO model) for scalability reasons\footnote{Requiring shards to remember the full history of inactive objects would increase their memory requirements monotonically over time, reaching at some point memory limits preventing further operations. Thus this is a poor mitigation for the attacks presented.}, and because \shard{3} takes part in only the second phase of \sbac.
The \attacker can double-spend $y_3$ repeatedly by replaying a single \prerecorded message multiple times, and spending the object (\ie purging it from \shard{3}'s UTXO) before each replay.

Contrarily to the attacks against the first phase of \sbac (\Cref{chainspace-attack-1}), these attacks do not rely on any racing conditions;  there is no need to race any honest messages.


\subsection{Real-world Impact}
\label{chainspace-discussion}
The real-world impact and attacker incentives to conduct these attacks depends on the nature and implementation of the smart contract handling the target objects. We discuss the impact of these attacks in the context of two common smart contract applications, which are also described in the \chainspace paper~\cite{chainspace}. To take a concrete example, we illustrate the attack depicted in \myrow{3} of \Cref{tab:sbac-attack}, but similar results can be obtained with the other attacks described in \Cref{tab:sbac-attack} and \Cref{tab:sbac-attack:2}.

One of the most common blockchain application is to manage cryptocurrency (or coins) and enable payments for processing transactions, implemented by the \textsf{CSCoin} smart contract in \chainspace.
Lets suppose object $x_1$ (handled by \shard{1}) represents Alice's account, and object $x_2$ (handled by \shard{2}) represents Bob's account. To transfer $v$ coins to Bob, Alice submits a transaction $T(x_1,x_2)\rightarrow(y_1,y_2)$, where $y_1$ and $y_2$ respectively represent the new account objects of Alice and Bob, with updated account balances. By executing the attack described in \myrow{3} of \Cref{tab:sbac-attack}, an \attacker can trick \shard{1} to abort the transaction and unlock $x_1$ (thus reestablishing Alice's account balance as it was prior to the coin transfer), and \shard{2} to accept the transaction and create $y_2$ (thus adding $v$ coins to Bob's account). This attack effectively allows any attacker to double-spend coins on the ledger; and shows how to create $v$ coins out of thin air.

Another common blockchain use case is a platform for decision making (or electronic petitions), implemented by the \textsf{SVote} smart contract in \chainspace. Upon initialization, the \textsf{SVote} contract creates two objects: \first $x_1$ representing the tally's public key, a list of all voters' public keys, and the tally's signature on these; and \second $x_2$ representing a vote object at the initial stage of the election (all candidates having a score of zero) along with a zero-knowledge proof asserting the correctness of the initial stage. To vote, clients submit a transaction $T(x_1,x_2)\rightarrow(y_1,y_2)$, where $y_1$ and $y_2$ are respectively the updated voting list (\ie, the voting list without the client's public key), and the election stage updated with the client's vote. By executing the attack described by \myrow{3} of \Cref{tab:sbac-attack}, an \attacker can trick \shard{1} to abort the transaction and thus not update the voting list, and \shard{2} to accept the transaction and thus update the election stage. This allows clients to vote multiple times during an election while remaining undetected (due to the privacy-preserving properties of the \textsf{SVote} smart contract).

\section{\Clientled Cross-shard Consensus} 
\label{client-led-consensus}

We describe replay attacks on \clientled cross-shard consensus protocols. We illustrate these attacks in the context of \omniledger~\cite{kokoris2017omniledger} (\Cref{omniledger}) to make the discussion concrete. However, we note that these attacks can be generalized to other similar systems. 
We discuss how the \attacker can record shard messages to replay in future attacks (\Cref{omniledger-message-recording}). We describe replay attacks on the first (Section~\ref{omniledger-attack-1}) and second (Section~\ref{omniledger-attack-2}) phase of the protocol. Finally, we discuss the real-world impact of these attacks (\Cref{omniledger-discussion}).

\subsection{\omniledger Overview} 
\label{omniledger}

\begin{figure}[t]
\centering
\includegraphics[width=.45\textwidth]{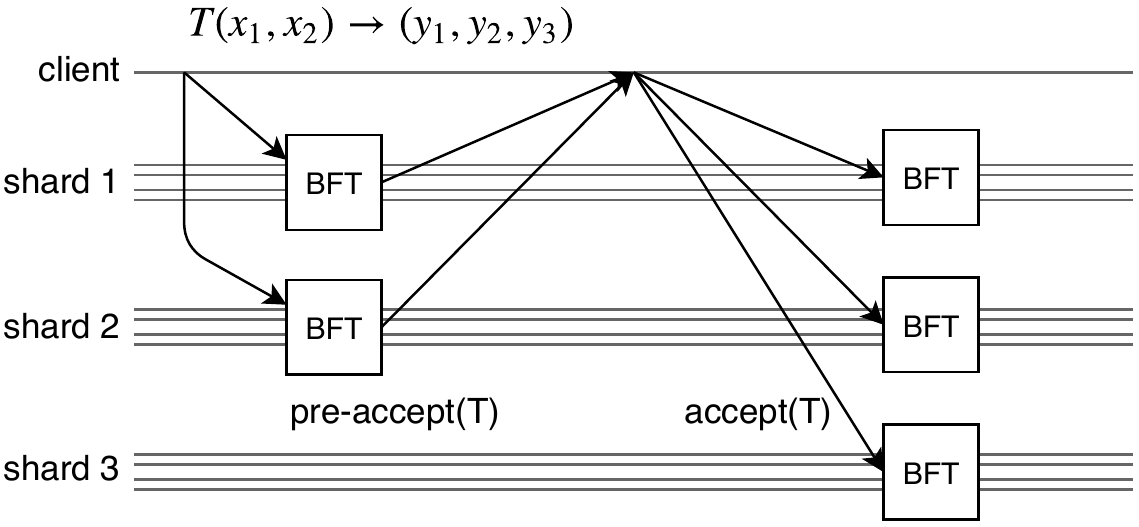}
\caption{\footnotesize 
An example execution of \atomix for a valid transaction $T(x_1,x_2)\rightarrow (y_1,y_2,y_3)$ with two inputs ($x_1$ and $x_2$, both are active) and three outputs $(y_1,y_2,y_3)$, where the final decision is \acceptt.}
\label{fig:atomix}
\end{figure}

\omniledger uses a \clientled cross-shard consensus protocol called \atomix.
The client submits the transaction $T$ to the input shards.
Each shard runs a BFT protocol locally to decide whether to accept or reject the transaction, and communicates its response (\preacceptt or \preabortt) to the client.\footnote{For consistency and clarity, we use the terminology used in \Cref{shard-led-consensus}. In \omniledger, \preacceptt is actually a \emph{proof-of-accept} and \preabortt is a \emph{proof-of-abort}~\cite{kokoris2017omniledger}.}  
A shard emits \preabortt if the transaction fails local checks.
Alternatively, if a shard emits \preacceptt, it inactivates the input objects it manages.
This is the first phase of \atomix, and is similar to the voting phase in the two-phase atomic commit protocol (\Cref{background}), but differs in that the protocol proceeds optimistically.
 The write changes made by the input shards in the first phase of \atomix are considered permanent (\ie, there is no \locked object state), unless the client requests the input shards to revert their changes in the second phase.  
After the client has collected \preacceptt from all input shards, it submits \acceptt message (containing proof of the \preacceptt messages) to the output shards which create the output objects.
Alternatively, if any of the input shards emits \preabortt, the client sends \abortt (containing proof of \preabortt) to the relevant input shards which make the input objects active again.
This is the second phase of \atomix, and is similar to the commit phase in the two-phase atomic commit protocol (\Cref{background}).

\Cref{fig:atomix} shows the execution of \atomix for a valid transaction $T(x_1,x_2)\rightarrow (y_1,y_2,y_3)$, with two active inputs ($x_1$ managed by \shard{1}, and $x_2$ managed by \shard{2}) and producing three outputs $(y_1,y_2,y_3)$ managed by \shard{1}, \shard{2} and \shard{3}, respectively. The client sends $T$ to the input shards, both of which reply with \preacceptt and make the input objects $x_1$ and $x_2$ inactive. The client sends \acceptt to the output shards which respectively create objects $y_1$, $y_2$, and $y_3$.



%
\begin{table*}[!ht]
\centering
\footnotesize
\begin{tabular}{c c c  c  c c c}
\toprule
\multicolumn{3}{c}{\small Phase 1 of \atomix} & & \multicolumn{3}{c}{\small Phase 2 of \atomix}\\
\noalign{\smallskip}
&\Centerstack{\textbf{Shard 1} \\ (potential victim)} & \Centerstack{\textbf{Shard 2} \\ (potential victim)} &
\Centerstack{\textbf{Client} \\ (victim)} &
\Centerstack{\textbf{Shard 1} \\ (potential victim)} & \Centerstack{\textbf{Shard 2} \\ (potential victim)} & \Centerstack{\textbf{Shard 3} \\ (potential victim)}\\
\noalign{\smallskip}
\toprule

\rowcolor{verylightgray}
1 & \Centerstack{\preacceptt \\ inactivate $x_1$} & \Centerstack{\preacceptt \\ inactivate $x_2$} & \acceptt & \Centerstack{- \\ create $y_1$} & \Centerstack{- \\ create $y_2$} & \Centerstack{- \\ create $y_3$} \\
\midrule
2 & $\rhd$ \preabortt & & \abortt & \Centerstack{- \\ re-activate $x_1$} & \Centerstack{- \\ re-activate $x_2$} & - \\
\midrule
3 & & $\rhd$ \preabortt & \abortt & \Centerstack{- \\ re-activate $x_1$} & \Centerstack{- \\ re-activate $x_2$} & - \\
\midrule
4 & $\rhd$\preabortt & $\rhd$\preabortt & \abortt & \Centerstack{- \\ re-activate $x_1$} & \Centerstack{- \\ re-activate $x_2$} & - \\
\noalign{\smallskip}
\toprule

\rowcolor{verylightgray}
5 & \Centerstack{\preabortt \\ -} & \Centerstack{\preacceptt \\ inactivate $x_2$} & \abortt & \Centerstack{- \\ -} & \Centerstack{- \\ re-activate $x_2$} & - \\
\midrule
6 & $\rhd$\preacceptt & & \acceptt & \Centerstack{- \\ create $y_1$} & \Centerstack{- \\ create $y_2$} & \Centerstack{- \\ create $y_3$} \\
\noalign{\smallskip}
\toprule

\rowcolor{verylightgray}
7 & \Centerstack{\preacceptt \\ inactivate $x_1$} & \Centerstack{\preabortt \\ -} & \abortt & \Centerstack{- \\ re-activate $x_1$} & \Centerstack{- \\ -} & - \\
\midrule
8 & & $\rhd$ \textsf{pre-accept}($T$) & \acceptt & \Centerstack{- \\ create $y_1$} & \Centerstack{- \\ create $y_2$} & \Centerstack{- \\ create $y_3$} \\
\noalign{\smallskip}
\toprule

\rowcolor{verylightgray}
9 & \Centerstack{\preabortt \\ -} & \Centerstack{\preabortt \\ -} & \abortt & \Centerstack{- \\ -} & \Centerstack{- \\ -} & - \\
\midrule
10 & $\rhd$ \textsf{pre-accept}($T$) & $\rhd$ \textsf{pre-accept}($T$) & \acceptt & \Centerstack{- \\ create $y_1$} & \Centerstack{- \\ create $y_2$} & \Centerstack{- \\ create $y_3$} \\
\bottomrule
\end{tabular}
\caption{\footnotesize List of replay attacks against the first phase of \atomix for all possible executions of a transaction $T(x_1,x_2)\rightarrow (y_1,y_2,y_3)$ as described in \Cref{attack-overview}. 
The highlighted rows indicate correct executions of \atomix (\ie, without the \attacker), and the other rows indicate incorrect executions due to the replay attacks. 
 In multirows, the top sub-rows show the protocol messages emitted by shards, and the bottom sub-rows indicate local shard actions as a result of emitting those messages.
  For example, (\mycolumn{1}, \myrow{1}) means that \shard{1} emits \preacceptt (top sub-row), and inactivates $x_1$ (bottom sub-row). 
The first two columns indicate the messages emitted by each shard at the end the first phase of \atomix. 
Replayed messages are marked with the symbol $\rhd$, for example $\rhd$\preabortt at (\mycolumn{1}, \myrow{2}) means that the \attacker sends to the client a \prerecorded \preabortt message impersonating \shard{1} that races the original \preacceptt (\mycolumn{1}, \myrow{1}) emitted by \shard{1}. 
 The third column indicates the messages sent by the client to the relevant shards, and the last three columns indicate the local actions performed by shards at the end of \atomix.
 }
\label{tab:atomix-attack}
\end{table*}

\subsection{Message Recording} 

\label{omniledger-message-recording}

Before launching the attacks, the \attacker first records the target shard responses. The \attacker can record shard responses in the first phase of \atomix (\ie, \preacceptt or \preabortt), enabling the attacks described in \Cref{omniledger-attack-1}. The \attacker can also record shard responses in the second phase of \atomix (\ie, \acceptt or \abortt), enabling the attacks described in \Cref{omniledger-attack-2}. 
In the general case, the \attacker passively collects the messages to replay, for example by protocol executions on the network, or by downloading the blockchain and selecting the appropriate messages. \Cref{sec:prerecoring-omniledger} shows how the \attacker can act as a client to actively elicit and record the target messages to later use in the replay attacks.


\subsection{Attacks on the First Phase of \atomix} 
\label{omniledger-attack-1}

We present replay attacks on the first phase of \atomix by taking the example of a transaction $T(x_1,x_2)\rightarrow (y_1,y_2,y_3)$ as described in \Cref{attack-overview}. These attacks easily generalize to transactions with $k$ inputs and $k'$ outputs managed by an arbitrary number of shards. The replay attacks work in two steps: \first the \attacker observes the traffic and records \preacceptt or \preabortt messages as described in \Cref{omniledger-message-recording}; and \second then replay those messages.

\Cref{tab:atomix-attack} shows the replay attacks that the \attacker can launch, for all possible combinations of responses generated by \shard{1} and \shard{2} in the first phase of \atomix. 
The caption includes details about how to interpret this table.
We describe \myrow{6} of \Cref{tab:atomix-attack}, to help readers interpret rest of the table on their own.
In the correct execution (\myrow{5}), \shard{1} emits \preabortt, and \shard{2} emits \preacceptt and inactivates the input objects $x_2$. 
Upon receiving these messages, the client sends \abortt to the output shards \shard{1}, \shard{2} and \shard{3}, and \shard{2} re-activates $x_2$; and the protocol terminates. 
In the attack illustrated in \myrow{6} of \Cref{tab:atomix-attack},
the \attacker races \shard{1} by sending to the client the \prerecorded \preacceptt message from \shard{1}. As a result, the client sends \acceptt message to the output shards \shard{1}, \shard{2} and \shard{3}, which respectively create the output objects $y_1$, $y_2$, and $y_3$.
As a result, the system ends up in an inconsistent state because the output objects ($y_1$, $y_2$, $y_3$) have been created, while the input object ($x_1$) was not active---this results in a double-spend of the input object $x_1$.

\Cref{tab:atomix-attack} shows that through careful selection of the messages to replay, the attacks can be effective against any shard. The attacks illustrated in \myrow{2}, \myrow{3}, and \myrow{4} only affect availability, while the other attacks compromise consistency (\ie, the \attacker can trick the input shards to reactivate arbitrary objects, and trick the output shards into creating new objects in violation of the protocol). The potential victims of these attacks include the client (\eg, when the \attacker replays the shard messages to it in the first phase of \atomix) and any input or output shards.

\subsection{Attacks on the Second Phase of \atomix} 
\label{omniledger-attack-2}

\begin{table*}[t]
\centering
\footnotesize
\begin{tabular}{c C C C C}
\toprule
\multicolumn{5}{c}{\small Phase 2 of \atomix}\\
& \bf Client & \Centerstack{\textbf{Shard 1} \\ (potential victim)} & \Centerstack{\textbf{Shard 2} \\ (potential victim)} & \Centerstack{\textbf{Shard 3} \\ (potential victim)}\\
\midrule

\rowcolor{verylightgray}
1 & \acceptt & \Centerstack{- \\ create $y_1$} & \Centerstack{- \\ create $y_2$} & \Centerstack{- \\ create $y_3$} \\
2 &  $\rhd$\abortt & \Centerstack{- \\ re-activate $x_1$} & \Centerstack{- \\ re-activate $x_2$} & - \\
\midrule

\rowcolor{verylightgray}
3 & \abortt  & \Centerstack{- \\ re-activate $x_1$} & \Centerstack{- \\ re-activate $x_2$} & - \\
4 & $\rhd$\acceptt & \Centerstack{- \\ create $y_1$} & \Centerstack{- \\ create $y_2$} & \Centerstack{- \\ create $y_3$} \\
\bottomrule
\end{tabular}
\caption{\footnotesize  List of replay attacks against the second phase of \atomix for all possible executions of the transaction $T(x_1,x_2)\rightarrow (y_1,y_2,y_3)$ as described in \Cref{attack-overview}. 
The highlighted rows indicate correct executions of \atomix (\ie, without the \attacker), and the other rows indicate incorrect executions due to the replay attacks.
In multirows, the top sub-rows show the protocol messages emitted by shards, and the bottom sub-rows indicate local shard actions.
Note that we use the multirow format for consistency reasons; in this table the first column indicates the messages emitted by the client  at the beginning of the second phase of \atomix, and the last two column shows the effect of these messages on the relevant shards.
Replayed messages are marked with the symbol $\rhd$.
For example, $\rhd$\abortt at (\mycolumn{1}, \myrow{2}) means that the \attacker sends a \prerecorded \abortt message to the input shards impersonating the client.}
\label{tab:atomix-attack:2}
\end{table*}

We present replay attacks on the second phase of \atomix. The \attacker \prerecords \acceptt and \abortt messages as described in \Cref{omniledger-message-recording} and \Cref{sec:prerecoring-omniledger}. 

\Cref{tab:atomix-attack:2} shows replay attacks corresponding to the messages emitted by the client in the second phase---\ie, \acceptt in \myrow{1}, or \abortt in \myrow{3}.
The caption includes details about how to interpret this table.
The \abortt message at (\mycolumn{1}, \myrow{2}) means that the \attacker sends a \prerecorded \abortt message to the input shards (\shard{1} and \shard{2}) impersonating the client. Upon receiving this message, \shard{1} and \shard{2} (incorrectly) re-activate $x_1$ and $x_2$, respectively. 
Furthermore, all output shards create the output objects when the correct \acceptt message emitted by the client (\myrow{1}, \mycolumn{1}) reaches them.
This results in inconsistent state, as the output objects are created, but the input objects are not consumed and are reactivated by the \abortt message replayed by the adversary. 
The potential victims of \abortt replay attack are the input shards.

Similarly, \acceptt at (\myrow{4}, \mycolumn{1}) means that the \attacker sends a \prerecorded \acceptt message to the output shards (\shard{1}, \shard{2} and \shard{3}) impersonating the client. Upon receiving this message, the output shards (incorrectly) create $y_1$, $y_2$ and $y_3$. 
Furthermore, the input shards (\shard{1} and \shard{2}) reactivate $x_1$ and $x_2$ upon receiving the correct \abortt message emitted by the client (\myrow{3}, \mycolumn{1}).
This creates inconsistent state: the input objects have not been consumed and have been reactivated by the \abortt message emitted by the client, but the output objects have been created due to the \acceptt message replayed by the \attacker. 
The potential victims of \acceptt replay attack are the output shards.

These attacks are possible because output shards create objects directly upon receiving \acceptt; they do not check if the objects have been previously invalidated because shards do not keep records of inactive objects (per the UTXO model) for scalability reasons.\footnote{Verifying that objects have not been previously invalided implies either keep a forever-growing list of invalidated objects, or download and check the shard's entire blockchain.}
The \attacker can double-spend the output objects repeatedly from a single \prerecorded message by replaying it multiple times, and spending the object (and effectively purging it from the output shards' UTXO) before each replay. 


Similar to the attacks against the second phase of \sbac (\Cref{chainspace-attack-2}), these attacks do not exploit any racing condition and can be mounted by an adversary at a leisurely pace.

\subsection{Real-world Impact}
\label{omniledger-discussion}
Contrarily to \chainspace, \omniledger does not support smart contracts and only handles a cryptocurrency. The attacks described in Sections~\ref{omniledger-attack-1}~and~\ref{omniledger-attack-2} allow an \attacker to: \first double-spend the coins of any user, by reactivating spent coins (\eg, the \attacker may execute the attack depicted by \myrow{2} of \Cref{tab:atomix-attack:2} to re-activate the objects $x_1$ and $x_2$ after the transfer is complete); and \second create coins out of thin air by replaying the message to create coins (\eg, an \attacker may execute the attack depicted by \myrow{4} of \Cref{tab:atomix-attack:2} to create multiple times object $y_3$, by purging it from the UTXO list of \shard{3} prior to each instance of the attack). 

If the \attacker colludes with the client, it can trigger the \prerecorded messages needed for the attacks as described in \Cref{omniledger-message-recording}.
Alternatively, the \attacker can passively observe the network and collect the target messages to replay.
Similar results can be obtained using the attacks described in \Cref{tab:atomix-attack}.
Note that since transaction are recorded on the blockchain, these attacks can be detected retrospectively.  
This can lead to the \attacker being exposed, or the \attacker can inculpate innocent users (the \attacker can replay messages of any user).


\section{Defenses Against Replay Attacks}

\label{defense}

We identify two issues that lead to the replay attacks described in \Cref{shard-led-consensus} and \Cref{client-led-consensus}, and discuss how to fix those:

\begin{itemize}

\item First, the input shards do not have a way to know that particular protocol messages received correspond to a specific instance (or session) of the protocol. This gap in the input shards' knowledge enables an \attacker to replay, mix and match, old messages leading to attacks.
To address this limitation, we associate a session identifier with each transaction, which has to be crafted carefully to not degrade the performance of the protocols significantly---such as, for example, by requiring nodes to store state linearly in the number of past transactions.

\item Second, in some cases the output shards are only involved in the second phase of the protocol, and therefore have no knowledge of the transaction context (to determine freshness) that is available to the input shards.
This limitation can be addressed by ensuring that all shards---input and output---witnesses the entire protocol execution, rather than just a subset of protocol execution phases. 

\end{itemize}

Note that the two mitigation techniques described above must be used together, as part of a single defense strategy against replay attacks. 
 
\section{The \sysname Protocol} 
\label{sec:byzcuit}

We showed that both \sbac (Sections~\ref{chainspace-attack-1}~and~\ref{chainspace-attack-2}) and \atomix (Sections~\ref{omniledger-attack-1}~and~\ref{omniledger-attack-2}) are vulnerable to replay attacks that can compromise system liveness and safety.
\atomix is the simpler protocol of the two, and using the client to coordinate cross-shard communication can reduce the cost to $O(n)$ in the number of shards (by aggregating shard messages).
However, an unresponsive or malicious client can permanently lock input objects by never initiating the second phase of the protocol, requiring additional design considerations (\eg, a new entity that periodically unlocks input objects for transactions on which no progress has been made). 
On the other hand, \sbac runs the protocol among the shards, without relying on client coordination.
 But this comes at the cost of increased cross-shard communication: all input shards communicate with all other input shards, which leads to communication complexity of $O(n^2)$ in the number of input shards.

Motivated by these insights, we present \sysname---a cross-shard atomic commit protocol (based on \sbac) that integrates design features from \atomix---and offers better performance and security against replay attacks. 
\sysname allocates a Transaction Manager (TM) to coordinate cross-shard communication, reducing its cost to $O(n)$ in the happy case\footnote{The communication complexity can be reduced to $O(n)$ in the number of shards by aggregating shard messages as described by \omniledger.}; alternatively \sysname also has a fall-back mode in case the TM fails, similar to \atomix and traditional two phase commit protocols. 

\sysname achieves resilience against the replay attacks described in \Cref{shard-led-consensus} and \Cref{client-led-consensus}, by leveraging the defense proposed in \Cref{defense}.


\subsection{\sysname Protocol Design} 
We describe how \sysname integrates the defense presented in \Cref{client-led-consensus}. To map particular protocol messages to a specific protocol instance (or session), \sysname associates a session identifier with each transaction. To ensure that all the relevant (input and output) shards witness all phases of the protocol execution, \sysname leverages the notion of \emph{dummy objects}:
each shard creates a fixed number of dummy objects upon configuration;
if a shard only serves as an output shard for a transaction (and therefore will only be involved in the second phase of the protocol), \sysname forces it to be involved in the first phase of the protocol by implicitly including a dummy object managed by the output shard in the transaction inputs, which will create a new dummy object upon completion.
As a result, the output shard also becomes an input shard (because of the inclusion of its dummy object in the transaction inputs) and witnesses the entire protocol execution, rather than just the second phase.

\para{\sysname Protocol Execution} We illustrate \sysname taking the example of a transaction $T(x_1,x_2)\rightarrow (y_1,y_2,y_3)$ with two input objects, $x_1$ managed by \shard{1} and $x_2$ managed by \shard{2}; and three outputs, $y_1$ managed by \shard{1}, $y_2$ managed by \shard{2}, and $y_3$ managed by \shard{3}.

\Cref{fig:sysname} illustrates the \sysname protocol; the client first sends the transaction to all input and output shards.
Note that this is different than other protocols like \sbac and \atomix, where the transaction is only sent to the input shards.
As mentioned previously, to achieve resilience against replay attacks, \sysname forces a shard that is \emph{only} involved in creating the output objects to also become an input shard (and witness the transnational context by participating in the first phase of the protocol) by implicitly consuming one of its dummy inputs (which creates a new dummy object upon completion). \sysname associates a sequence number $s_{x_i}$ to each object and dummy object (when the object is created $s_{x_i}=0$). The sequence number is intrinsically linked to the object: when clients query shards to obtain an object $x_i$, they also receive the associated sequence number $s_{x_i}$. 

When submitting the transaction $T$, the client also sends along a transaction sequence number $s_{T}=max\{ s_{x_1}, s_{x_2}, s_{d_3} \}$,
where the transaction sequence number $s_{T}$ is the maximum of the sequence numbers $s_{x_i}$ of each input object $x_i$ and dummy objects $d_i$~(\ding{202}). 

Upon receiving a new pair $(T,s_T)$, each shard saves $(T,s_T)$ in a local cache memory---the transaction sequence number $s_T$ acts as session identifier associated with the transaction $T$. Each shard internally verifies that the transaction passes local checks, and that $s_T$ is equal to (or bigger than) the sequence numbers of the objects they manage (\ie, \shard{1} checks $s_T \geq s_{x_1}$, \shard{2} checks $s_T \geq s_{x_2}$, \shard{3} checks $s_T \geq s_{d_3}$). The shards send their local decision to the TM: \preacceptts for local accept (and the shard locks the objects it manages),
or \preabortts for local abort.

After receiving all the messages corresponding to the first phase of \sysname from the concerned shards, the TM sends a suitable message to the shards (\acceptts if all the shards respond with \preacceptts, or \abortts otherwise). 
Upon receiving \acceptts or \abortts from the TM, shards first verify that they previously cached the pair $(T, s_T)$ associated with the message; otherwise they ignore it~(\ding{203}).

The \acceptts or \abortts messages sent by the TM provide enough evidence to the shards to verify whether $s_T$ is correctly computed; \ie shards verify that $s_T$ is at least the maximum of the sequence numbers of each input and dummy object by inspecting the transaction $T$ signed by each shard. If \acceptts has a correct $s_T$, the shards inactivate the input objects and create the output objects $(y_1,y_2,y_3)$, and \shard{3} creates a new dummy object $\widetilde{d}_{3}$;  otherwise, they update the sequence numbers of each input object $(s_{x_1},s_{x_2})$ and dummy object $d_{3}$ to $(s_T+1)$, \ie shards locally update $s_{x_1} \leftarrow (s_T+1)$ and $s_{x_2} \leftarrow (s_T+1)$, and $s_{d_{3}}\leftarrow (s_T+1)$. Shards delete $(T,s_T)$ from their local cache~(\ding{204}).

Since we assume that shards are honest---inline with the threat model of the systems discussed---it suffices if only one shard notifies the client of the protocol outcome; we may set any arbitrary rule to decide which shard notifies the client (\eg, the shard handling the first input object)~(\ding{205}).
\Cref{fig:byzcuit-fsm} shows the finite state machine describing the life cycle of \sysname objects. 

%
\begin{figure}[t]
\centering
\includegraphics[width=.45\textwidth]{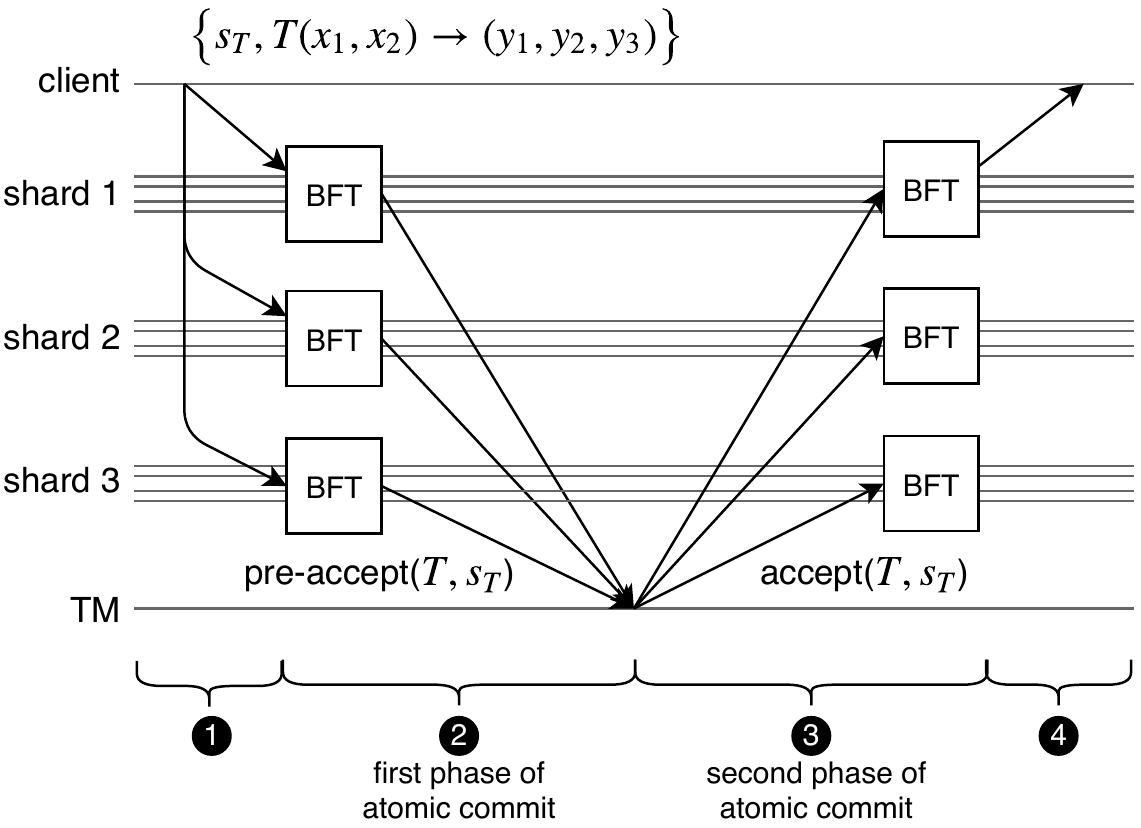}
\caption{\footnotesize An example execution of \sysname for a valid transaction $T(x_1,x_2)\rightarrow (y_1,y_2,y_3)$ with two input objects ($x_1$ and $x_2$, both are active), and three outputs $(y_1,y_2,y_3)$, where the final decision is \acceptts.}
\label{fig:sysname}
\end{figure}

\para{Transaction Manager} 
The Transaction Manager (TM) coordinates cross-shard communication in \sysname. We now discuss who might play the role of the TM, and argue that \sysname guarantees liveness even if the TM is faulty (byzantine) or crashes.

Keeping with the overall design goal of decentralization, we envision that a designated shard will act as the TM. If the shard is honest, the TM is live---and therefore progress is always made. The input shards contact in turn each node of the TM shard until they reach one honest node. The TM shard may have up to $f$ dishonest nodes; therefore, the client or the input shards need to send messages to at least $f+1$ nodes of the TM shard to ensure that it is received by at least one honest node\footnote{Clients may take a statistical view of availability. Given that fewer than $2/3$ of nodes in a shard are dishonest, sending the transaction to $\rho$ nodes will fail to reach an honest node with probability only $(1/3)^{\rho}$. Clients may send messages sequentially to nodes, and only continue if they do not observe progress within some timeout to further reduce costs.}. Thus, as soon as the first honest node receives the message, the protocol progresses.

If the TM is the client or any centralized party, it may act arbitrarily---but this does not stall the protocol because anyone can make the protocol progress by taking over at any time the role of the TM. This is possible because the TM does not act on the basis of any secrets, therefore anyone else can take over and complete the protocols. This ``anyone'' may be an honest node in a shard that wants to finally unlock an object (\eg, upon a timeout); or other clients that wish to use a locked object; or it may be an external service that has a job to periodically close open \sysname instances. \sysname ensures such parties may attempt to make progress asynchronously and concurrently safely.
As a result, \sysname guarantees liveness as long as there is at least one honest entity in the system, willing to act as TM and drive the protocol.

\begin{figure}[t]
\centering
\includegraphics[width=.48\textwidth]{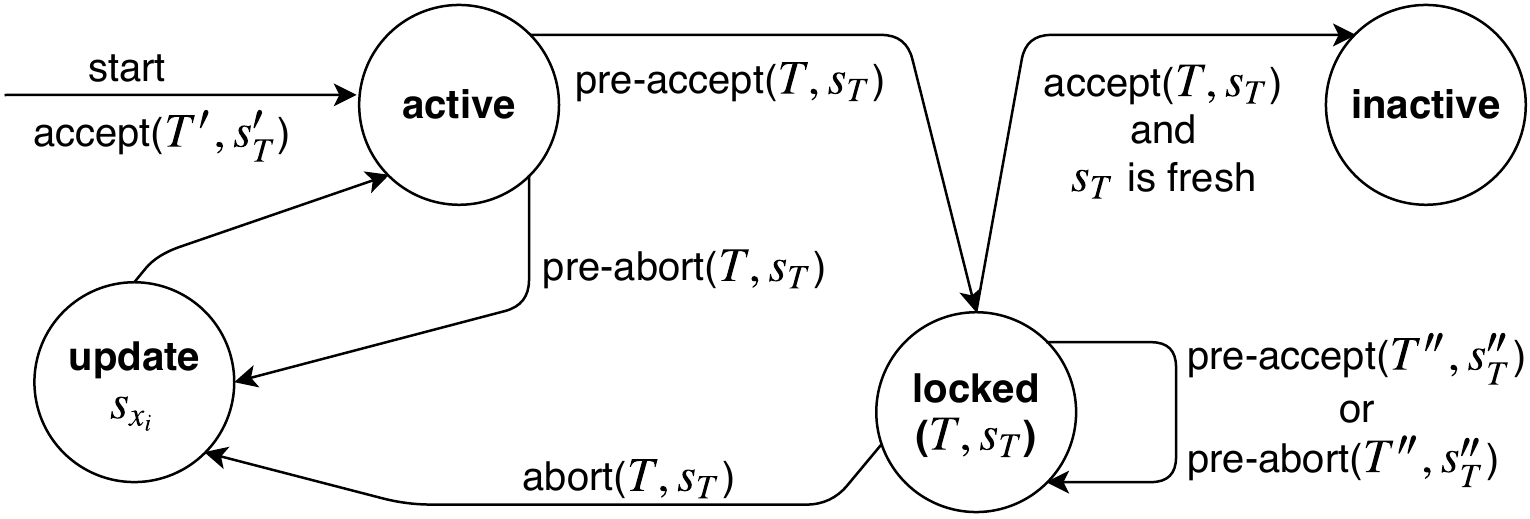}
\caption{\footnotesize State machine representing the life cycle of objects in \sysname. Objects are initially \activeObj. Upon receiving a transaction that passes local checks, a shard changes its input objects' state to \locked  (objects are locked for a given transaction $T$ and transaction sequence number $s_T$) and emits \preacceptts; otherwise it updates the sequence number of every object it manages and emits \abortts. 
Once a shard locks input objects for a given $(T,s_T)$, any \acceptts and \abortts with malformed transaction sequence numbers are ignored, and do not cause any transition (not included in the figure). Any incoming transaction $T'$ that requires processing \locked input object(s) is aborted.
Upon receiving \acceptts with a well formed $s_T$, a shard makes its input objects \inactiveObj and creates the output objects.
Alternatively, upon receiving \abortts shards unlock their inputs and updates the corresponding sequence numbers.
}
\label{fig:byzcuit-fsm}
\end{figure}

\para{Handling Sequence Number Overflow}
An \attacker can try to exhaust the possible sequence numbers to make them overflow. The \attacker submits a pair $(T,s_T)$ such that the sequence number $s_T$ is just below the system overflow value; the sequence numbers associate with the inputs overflow upon the next updates, and the system would be again prone to the attacks described in \Cref{chainspace-attack-1}\footnote{Note that this overflow vulnerability is common to every system relying on nonces chosen by the users, like Byzantine Quorum Systems~\cite{malkhi1998byzantine}.}. To mitigate this issue, shards define a \emph{clone} procedure allowing to update any of their objects to an unchanged version of themselves (\ie it creates a fresh copy of the object). This clone procedure effectively creates a new object with serial number $s_x'=0$. When shards detect that the serial number of one of their objects approaches the overflow value, they execute internally this clone procedure. The \attacker may exploit this mechanism to DoS the system, forcing shards to constantly update their objects; as a result, the target objects are not available to users. DoS countermeasures are out of scope, and are typically addressed by introducing transaction fees.

\subsection{Security against Replay Attacks}\label{sec:security}
We argue that \sysname is resilient to replay attacks. We recall the Honest Shard assumption from \chainspace and \omniledger under which \sysname operates, and assume that messages are authenticated as in traditional BFT protocols. 

\begin{SecAssumption}\label{def:honest-shard}(Honest Shard~\cite{chainspace})
The adversary may create arbitrary smart contracts, and input arbitrary transactions into \sysname, however they are bound to only control up to $f$ faulty nodes in any shard. As a result, and to ensure the correctness and liveness properties of Byzantine consensus, each shard must have a size of at least $3f + 1$ nodes. (From \chainspace~\cite{chainspace}.)
\end{SecAssumption}

Any message emitted by shards comes with at least $f+1$ signatures from nodes. Assuming honest shards, the \attacker can forge at most $f$ signatures, which is not enough to impersonate a shard. We use the Lemma below to prove the security of \sysname.

\begin{lemma} \label{le:secure-phase-1}
Under Honest Shard assumption, no \attacker can obtain \prerecorded messages containing a fresh transaction sequence number $s_T$.
\end{lemma}
\begin{proof}
The core idea protecting \sysname from these replay attacks is that the \attacker can only obtain \prerecorded messages associated with old transaction sequence numbers $s_T$. The transaction sequence number $s_T$ is fresh only if it is at least equal the maximum of the sequence number of all input and dummy objects of the transaction $T$. Shards update every input and dummy object sequence number upon aborting transactions in such a way that sequence numbers only increase. That is, after emitting \preacceptts or \preabortts, either the sequence number of all input and dummy objects of $T$ are updated to a value bigger than $s_T$ (in case of \preabortts), or the objects are inactivated which prevents any successive transaction to use them as input (in case of \preacceptts).
It is therefore impossible for the adversary to hold a \prerecorded message for a fresh $s_T$ since the only \prerecorded messages that the adversary can obtain contain sequence numbers smaller than $s_T$.
\end{proof}

\para{Security of the first phase of \sysname}
An \attacker may try to replay \preacceptts and \preabortts during the first phase of the protocol, similarly to the attacks described in Sections~\ref{chainspace-attack-1}~and~\ref{omniledger-attack-1}; the TM then aggregates these messages into either \acceptts or \abortts, and forwards them to the shards during the second phase of the protocol. 

\Cref{th:secure-phase-1} shows that \sysname detects that they originate from replayed messages and ignores them. Intuitively, the transaction sequence number $s_T$ acts as a monotonically increasing session identifier associated with the transaction $T$; the \attacker cannot obtain \prerecorded messages containing a fresh $s_T$. \sysname shards can then distinguish replayed messages (\ie, messages with old $s_T$) from the messages of the current instance of the protocol (\ie, messages with fresh $s_T$).

\begin{theorem} \label{th:secure-phase-1}
Under Honest Shard assumption, \sysname ignores \acceptts and \abortts messages issued from replayed \preacceptts and \preabortts.
\end{theorem}
\begin{proof}
\Cref{fig:byzcuit-fsm} shows that once \sysname locks objects for a particular pair $(T,s_T)$, the protocol can only progress toward \acceptts or \abortts; \ie shards can either accept or abort the transaction $T$. The \attacker aims to trick one or more shards to incorrectly accept or abort $T$ by injecting \prerecorded messages during the first phase of \sysname; we show that the \attacker fails in every scenario. 

Suppose that a transaction $T$ should abort (the TM outputs \abortts), but the \attacker tries to trick some shards to accept the transaction. \Cref{fig:byzcuit-fsm} shows that the \attacker can only succeed the attack if they gather \acceptts containing a fresh transaction sequence number $s_T$. Lemma~\ref{le:secure-phase-1} states that no \attacker can obtain \prerecorded messages over a fresh transaction sequence number $s_T$; therefore the only messages available to the adversary at this point of the protocol are (at most) $n-1$ \preacceptts and (at most) $n$ \abortts, where $n$ is the number of concerned shards. This is not enough to form an \acceptts message with a fresh transaction sequence number $s_T$ (which is composed of $n$ \preacceptts messages); therefore the \attacker cannot trick any shard to accept the transaction.

Suppose that a transaction $T$ should be accepted (the TM outputs \acceptts with a fresh $S_T$), but the \attacker tries to trick some shards to abort the transaction. \Cref{fig:byzcuit-fsm} show that \sysname does not require a fresh transaction sequence number $s_T$ to abort transactions (the freshness of $s_T$ is only enforced upon accepting a transaction); but shards locked the input and dummy objects of the transaction for the pair $(T,s_T)$ (with fresh $s_T$), so the attacker needs to gather \abortts containing the same transaction sequence number $s_T$ locked by shards.  Lemma~\ref{le:secure-phase-1} shows that the \attacker cannot obtain \prerecorded messages over fresh $s_T$; therefore the only messages available to the adversary containing the (fresh) $s_T$ locked by shards at this point of the protocol are $n$ \preacceptts. This is not enough to form an \abortts message (which is composed of at least one \preabortts); therefore the \attacker cannot trick any shard to abort the transaction.
\end{proof}
%


\para{Security of the second phase of \sysname}
An \attacker may try to replay \acceptts and \abortts messages during the second phase of the protocol, similarly to the attacks described in Sections~\ref{chainspace-attack-2}~and~\ref{omniledger-attack-2}.

\Cref{th:secure-phase-2} shows that \sysname ignores those replayed messages. Intuitively, these attacks target shards acting only as output shards (and not also as input shards) and exploit the fact that they are only involved in the second phase of the protocol, and therefore have no knowledge of the transaction context (to determine freshness) that is available to the input shards. \sysname is resilient to these replay attacks as it is designed in such a way that there are no shards that act only as output shards; all output shards are forced to also become input shards, by introducing dummy objects if they do not manage any input objects; this prevents the attacks by removing the attack target.

\begin{theorem} \label{th:secure-phase-2}
Under Honest Shard assumption, \sysname ignores replayed \acceptts and \abortts messages.
\end{theorem}
\begin{proof}
\Cref{fig:byzcuit-fsm} shows that shards only act upon \acceptts and \abortts messages if they have the pair $(T,s_T)$ saved in their local cache\footnote{Contrarily to \sbac and \atomix, all \sysname shards have the pair $(T,s_T)$ in their local cache after as they all participate to the first phase of the protocol.}. Shards save a pair $(T,s_T)$ in their local cache upon emitting \preacceptts or \preabortts, and delete it at the end of the protocol; therefore the only attack windows where the adversary can replay \acceptts and \abortts messages is while the transaction $T$ (associated with $s_T$) is being processed by the second phase of the protocol. This forces the attacker to operates under the same conditions as \Cref{th:secure-phase-1}.
\end{proof}
%

Appendix~\ref{properties} shows that \sysname guarantees liveness, consistency and validity, similarly to \sbac.

\section{Implementation \& Evaluation} \label{sec:implementation}

We implement a prototype of \sysname (Section~\ref{sec:byzcuit}) in Java and evaluate its performance and scalability. To analyze the overhead introduced by our replay attack defenses (\ie, with message sequence numbers and dummy objects), we compare \sysname with replay defenses (\sysnamereplay) with the baseline of \sysname without any replay attack defenses (\sysnamebaseline). 

Our implementation of \sysname is a fork of the \chainspace code~\cite{chainspace}, and is released as an open-source project\footnote{\url{https://github.com/sheharbano/byzcuit}}.
For BFT consensus, we use the \bftsmart~\cite{bftsmart} Java library (based on PBFT~\cite{pbft}), which is one of the very few maintained open source BFT libraries. End users run a client to communicate with \sysname nodes, which sends transactions according to the \bftsmart protocol. The \sysname client also acts as the Transaction Manager (TM) and is responsible for driving the cross-shard consensus.

We evaluate the performance and scalability of our \sysname implementation through deployments on Amazon EC2 containers. 
We also compare \sysname with \chainspace to measure performance improvements, by running our evaluations in a similar setup as \chainspace. 
We launch up to 96 instances for shard nodes and 96 instances for clients on \emph{t2.medium} virtual machines, each containing 8 GB of RAM on 2 virtual CPUs and running GNU/Linux Debian 8.1. We use 4 nodes per shard. 
Each measured data point corresponds to 10 runs represented by error bars. The error bars in \Cref{fig:tpsVSshards} and \Cref{fig:tpsVSdummy} show the average and standard deviation, and the error bars in \Cref{fig:latencyVStps} show the median and the 75th and 25th percentiles.

\begin{figure}[t]
\centering
\includegraphics[width=.45\textwidth]{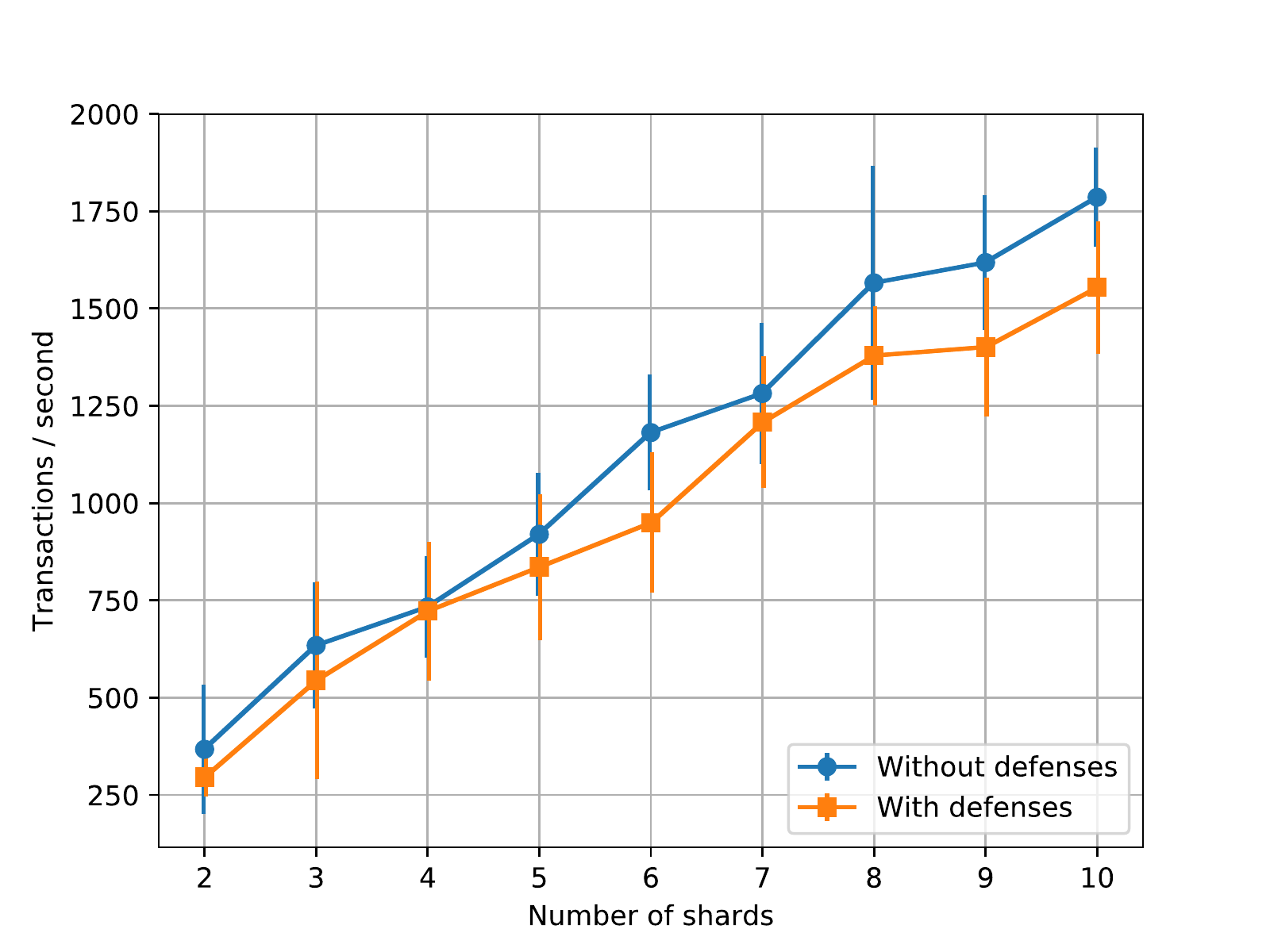}
\caption{\footnotesize The effect of the number of shards on throughput. Each transaction has 2 input objects and 5 output objects, both chosen randomly from shards.}
\label{fig:tpsVSshards}
\end{figure}
\begin{figure}[t]
\centering
\includegraphics[width=.45\textwidth]{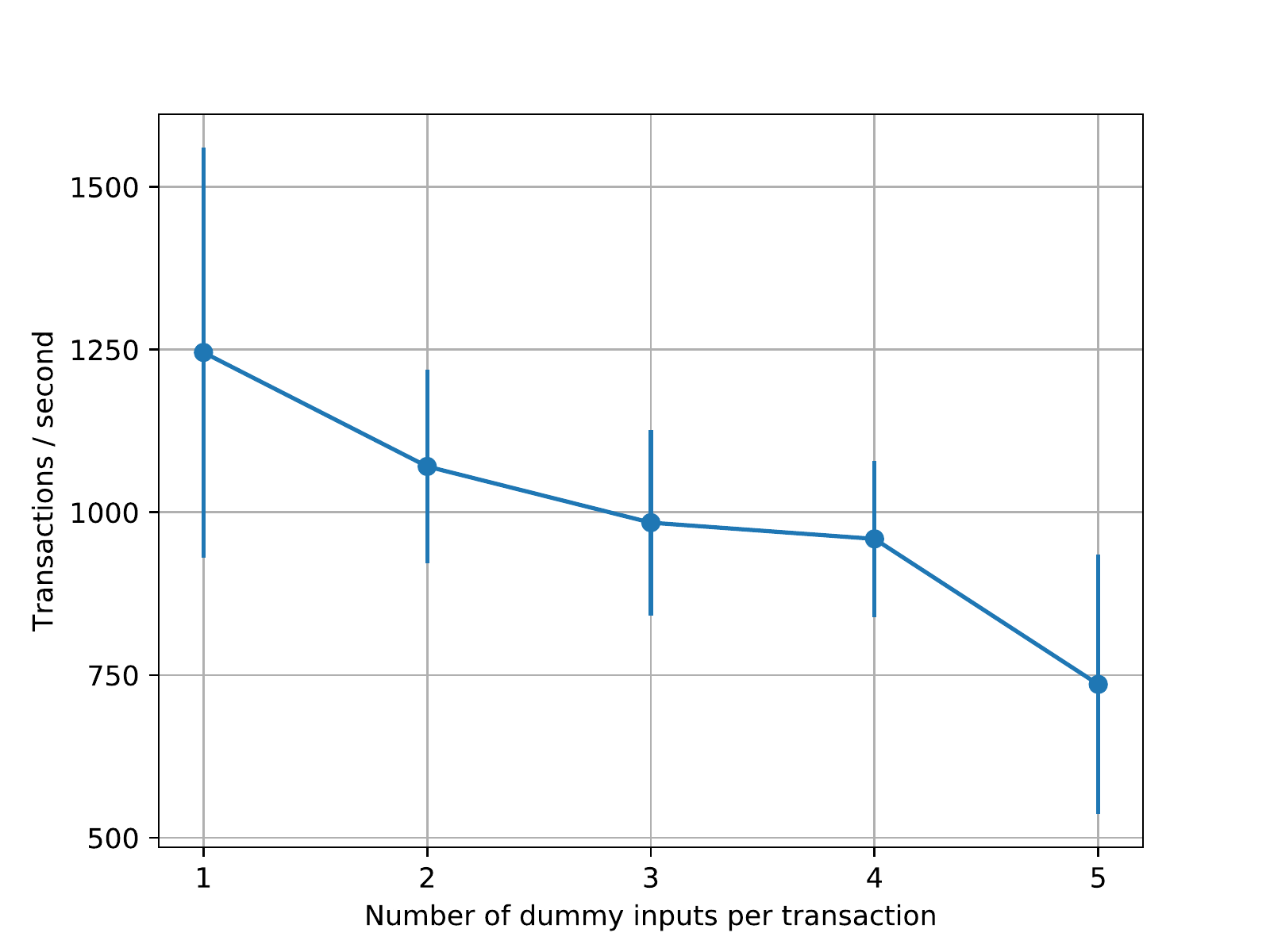}
\caption{\footnotesize Decrease of \sysname throughput with the number of dummy objects. Each transaction has 1 input object, and up to 5 dummy objects randomly selected from unique non-input shards. 6 shards are used.}
\label{fig:tpsVSdummy}
\end{figure}
\begin{figure}[t]
\centering
\includegraphics[width=.45\textwidth]{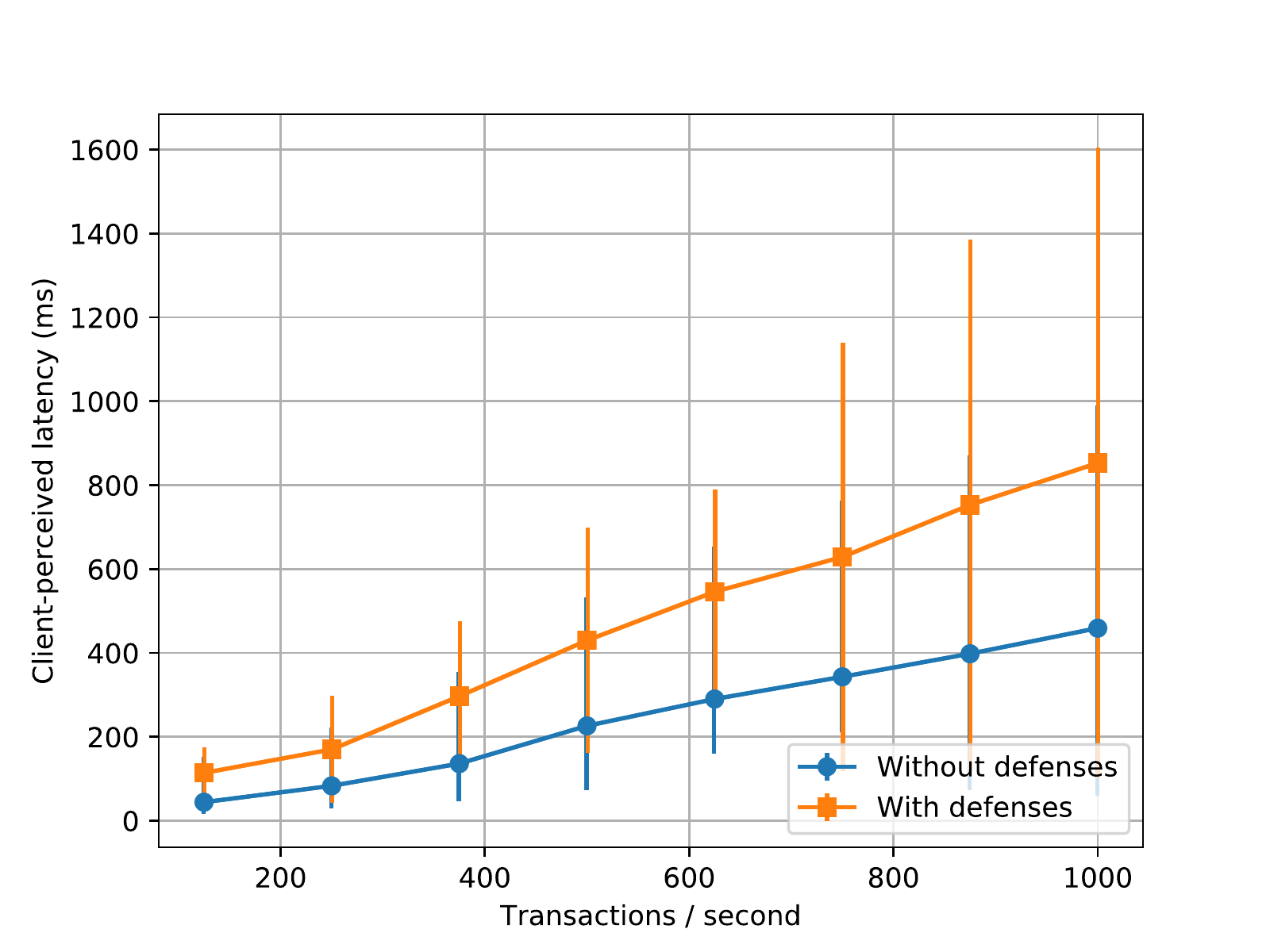}
\caption{\footnotesize Client-perceived latency vs. system load (number of transactions received per second by \sysname), for 6 shards with 2 inputs and 5 outputs per transaction (both chosen randomly from shards).}
\label{fig:latencyVStps}
\end{figure}

\mypara{Throughput and Scalability}
\Cref{fig:tpsVSshards} shows the throughput of \sysname (the number of transactions processed per second, tps) corresponding to an increasing number of shards. 
Each transaction has 2 input objects and 5 output objects, chosen randomly from shards. We test transactions with 5 output objects for a fair evaluation of \sysname's replay defenses by triggering the creation of dummy objects (\ie, a large number of output objects and a small number of input objects implies a higher probability of output-only shards getting selected, triggering the creation of dummy objects). We find that \sysnamereplay has a throughput of 260 tps for 2 shards, and linearly scales with the addition of more shards achieving up to 1550 tps for 10 shards. As expected, the throughput of \sysnamereplay is lower than \sysnamebaseline by a somewhat constant factor ranging from 20--200 tps, but still increases linearly. This is expected because the creation of dummy objects in \sysnamereplay leads to a higher number of shards processing the same transaction compared to \sysnamebaseline, leading to lower concurrency and lower throughput. 

Another interesting observation is that the design and implementation optimizations in \sysnamereplay lead to significantly higher throughput than \chainspace, even though the former has lower concurrency due to the dummy objects. For similar experimental setup and for 2--10 shards, \chainspace achieves 70--180 tps, while \sysnamereplay achieves 260--1550 tps. This is due to the improved design of the cross-shard consensus protocol (Section~\ref{sec:byzcuit}), which results in communication complexity of $O(n)$ in contrast to \chainspace's $O(n^2)$ (where $n$ is the number of input shards). Another reason for \sysnamereplay's significant throughput improvement is that unlike \chainspace, all interactions between the Transaction Manager and the shards are asynchronous. This eliminates the blocking condition in \chainspace where a shard cannot commit a transaction in the second phase of the cross-shard consensus protocol, until it receives messages from all concerned shards corresponding to the first phase.  

\mypara{The Effect of Dummy Objects on Throughput}
We previously observed that dummy objects reduce the throughput of \sysnamereplay with respect to \sysnamebaseline. \Cref{fig:tpsVSdummy} shows the extent of throughput degradation due to dummy objects. We submit specially crafted transactions to 6 shards, such that each transaction has 1 input object, and we vary the number of dummy objects from 1--5 selected from unique shards, resulting in a corresponding decrease in concurrency because as many shards end up processing the transaction. For example, 2 dummy objects means that 3 shards process the transaction (1 input shard, and 2 more shards corresponding to the dummy objects).
As expected, the throughput decreases by 20--250 tps with the addition of each dummy object, and reaches 750 tps when all 6 shards handle all transactions.

\mypara{Client-perceived Latency}
\Cref{fig:latencyVStps} shows the client-perceived latency---the time from when a client submits a transaction, until it receives a decision from \sysname about whether the transaction has been committed---under varying system loads (expressed as transactions submitted to \sysname per second). We submit a total of 1200 transactions at 200--1000 transactions per second to \sysname with 6 shards. Each transaction has 2 inputs objects and 5 output objects, both chosen randomly from shards. When the system is experiencing a load of up to 1000 tps, clients hear back about their transactions in less than a second on average, even with our replay attack defenses.

\section{Comparison with Mutex-based Cross-shard Consensus Protocols} \label{sec:ethereum}
Mutex-based schemes for cross-shard transactions, such as \ethereum's cross-shard ``yanking'' proposal~\cite{yanking}, find a way to avoid complex cross-shard coordination for transactions that involve objects managed by different shards.
The key idea is to require all objects
 that a transaction reads or writes 
 to be in the same shard (\ie, all locks for a transaction are local to the shard). Cross-shard transactions are enabled by transferring the concerned objects between shards, effectively giving shards a lock on those objects. 
 When \shard{1} transfers an object to \shard{2}, \shard{1} includes a transfer ``receipt'' in its blockchain. 
  A client can then send to \shard{2} a Merkle proof of this receipt being included in \shard{1}'s blockchain, which makes the object active in \shard{2}.

Mutex-based schemes also need to consider replay attacks.
Clients can claim the same receipt multiple times, unless shards store information about previously claimed receipts. 
Na\"{\i}vely, shards have to store information about all previously claimed receipts permanently.
However, two intermediate options with trade-offs have been proposed~\cite{yanking}:
\begin{itemize}
    \item Shards only store information about receipts for $l$ blocks; so clients can only claim receipts within $l$ blocks, and objects are permanently lost if not claimed within $l$ blocks.
    \item Shards only store information about receipts for $l$ blocks, and include the root of a Merkle tree of claimed receipts in their blockchain every $l$ blocks. If a receipt is not claimed within $l$ blocks, the client must provide one Merkle proof every $l$ blocks that have passed to show that the receipt has not been previously claimed, in order to claim it. The longer the receipt was not claimed, the greater the number of proofs that are needed to claim a receipt. These proofs need to be also stored on-chain to allow other nodes to validate them.
\end{itemize}
\sysname forgoes the need for shards to store information about old state (such as inactive objects or old receipts) as shards only need to know the set of active objects they manage, and does not impose a trade-off between the amount of information about old state that needs to be stored and the cost of recovering old state that was held up in an incomplete cross-shard transaction (\ie, an unclaimed receipt).
\section{Conclusion} 

\label{sec:conclusion}

We presented the first replay attacks against cross-shard consensus protocols in sharded distributed ledgers. These attacks affect both shard-driven and client-driven consensus protocols, and allow attackers to double-spend or lock objects with minimal efforts. The attacker can act independently without colluding with any nodes, and succeed even if all nodes are honest; most of the attacks work without making any assumptions on the underlying network. While addressing these attacks seems like an implementation detail, their many variants illustrate that a fundamental re-think of cross-shard commit protocols is required to protect against them.

We developed \sysname, a new cross-shard consensus protocol merging features from \shardled and \clientled consensus protocols, and withstanding replay attacks. \sysname can be seen as unifying Atomix (from \omniledger) and \sbac (from \chainspace), into an $O(n)$ protocol, that is efficient and secure. We implemented a prototype of \sysname and evaluated it on a real cloud-based testbed, showing that it is more efficient than \chainspace, and on par with \omniledger performance. The resulting protocol is a drop-in replacement for either, and can be adopted to immunize systems based on those designs.

\section*{Acknowledgements} At the time of this work, George Danezis, Shehar Bano and Alberto Sonnino were supported in part by EPSRC Grant EP/N028104/1 and the EU H2020 DECODE project under grant agreement number 732546 as well as \texttt{chainspace.io}. Mustafa Al-Bassam is supported by The Alan Turing Institute. We thank Eleftherios Kokoris-Kogias for helpful suggestions on early manuscripts.
We appreciate the valuable feedback we
received from our shepherd and the anonymous
reviewers.


{\normalsize \bibliographystyle{acm}
\bibliography{references}}


\appendices
\section{Eliciting Messages to Replay} \label{sec:prerecoring}
This appendix shows how the \attacker can act as (or collude with) a client to actively elicit and record the target messages to later use in the replay attacks. This empowers the \attacker to actively orchestrate the attacks.

We describe how the \attacker can trigger target messages in the context of an example, without loss of generality.
Lets assume that \shard{1} manages objects $x_{1}$ (\activeObj) and object $\widetilde{x_1}$ (\inactiveObj or non-existent), and \shard{2} manages object $x_{2}$ (\activeObj);
$\widetilde{x*}$ means any inactive object on the shard, and $y*$ means any output object (\ie, their details do not matter).

\subsection{\Shardled Cross-Shard Consensus} \label{sec:prerecoring-chainspace}
We show how the \attacker can act as (or collude with) a client to actively elicit and record the target messages, in the context of \shardled cross-shard consensus protocols as illustrated by \Cref{shard-led-consensus}.
To elicit \preacceptt for a transaction $T(x_{1},x_{2})\rightarrow (y*)$ (the output $y*$ is not relevant here) from \shard{1}, the key consideration is to closely precede the transaction with another transaction $T'$ that: \first locks the inputs managed by at least one other shard (in this case $x_2$ on \shard{2}); and \second to ensure that the preceding transaction $T'$ gets ultimately aborted, and $x_2$ becomes active again. The steps look as follows:  


\begin{itemize}
    \item The \attacker submits $T'(x_{2},\widetilde{x*})\rightarrow (y*)$ to \shard{2}. This locks $x_2$.
    \item The attacker quickly follows up by submitting $T(x_{1},x_{2})\rightarrow (y*)$ to \shard{1} and \shard{2}. \Shard{1} generates \preacceptt, which is the target message that the attacker records. \Shard{2} generates \preabortt because $x_2$ is locked by $T'$. Consequently, in the second phase of \sbac, both \shard{1} and \shard{2} end up aborting $T$.
    \item $T'$ is eventually aborted, making $x_{2}$ active again.
\end{itemize}


To elicit \preabortt for a transaction $T(x_{1},x_{2})\rightarrow (y*)$ (the output $y*$ is not relevant here) from \shard{1}, the key consideration is to closely precede the transaction with another transaction $T'$ that locks the input managed by the shard (in this case $x_1$ on \shard{1}). The steps look as follows:
\begin{itemize}
    \item The attacker submits  $T'(x_{1},\widetilde{x*})\rightarrow (y*)$ to \shard{1}. This locks $x_1$.
    \item The attacker quickly follows up by submitting $T(x_{1},x_{2})\rightarrow (y*)$ to \shard{1} and \shard{2}. \Shard{1} generates \preabortt because $x_1$ is locked by $T'$, which is the target message that the attacker records. \Shard{2} generates \preacceptt. Consequently, in the second phase of \sbac, both \shard{1} and \shard{2} end up aborting $T$.
     \item $T'$ is eventually aborted, making $x_{1}$ active again.
\end{itemize}

To elicit \acceptt used by the attacks described in \Cref{chainspace-attack-2}, the \attacker simply submits transaction $T$ and observes and records its successful execution. The \attacker has no incentive to record \abortt messages as these are ignored by shards (see \Cref{tab:sbac-attack:2}).

\subsection{\Clientled Cross-Shard Consensus} \label{sec:prerecoring-omniledger}
We show how the \attacker can act as (or collude with) a client to actively elicit and record the target messages, in the context of \clientled cross-shard consensus protocols as illustrated by \Cref{client-led-consensus}.
 
To elicit \preacceptt from \shard{1} for a transaction $T(x_{1},x_{2})\rightarrow (y*)$ (the output $y*$ is not relevant here) from \shard{1}, the key consideration is to closely precede the transaction with another transaction that: \first temporarily spends the inputs managed by at least one other shard (in this case $x_2$ on \shard{2}); and \second to ensure that the preceding transaction is ultimately aborted so that $x_2$ becomes active again. The steps look as follows:

\begin{itemize}
    \item The \attacker submits $T'(x_2,\widetilde{x*})\rightarrow (y*)$ to \shard{2}, where $\widetilde{x*}$ is managed by a different shard. \Shard{2} emits \preaccepttt and marks $x_2$ as inactive.
    \item The \attacker follows up by submitting $T(x_{1},x_{2})\rightarrow (y*)$ to \shard{1} and \shard{2}. \Shard{1} generates \preacceptt, which is the target message that the attacker records. \Shard{2} generates \preabortt because $x_2$ is inactive.
    \item The \attacker submits \abortt to \shard{1} to reactivate $x_1$, and sends \aborttt to \shard{2} to reactivate $x_2$.
\end{itemize}

For the attacks described in \Cref{omniledger-attack-2}, the \attacker needs to elicit \abortt and \acceptt from the target shards. For the former, the \attacker can follow the steps described previously to elicit \preacceptt and \preabortt. To elicit \acceptt, the \attacker simply submits transaction $T$ and observes and records its successful execution. 

\section{\sysname Security \& Correctness}\label{properties}
We show that \sysname guarantees liveness, consistency, and validity similarly to \sbac.

\begin{theorem}
(Liveness~\cite{chainspace}) Under Honest Shards assumption, a transaction $T$ that is proposed to at least one honest concerned node, eventually results in either being committed or aborted, namely all parties deciding \acceptts or \abortts. (From \chainspace~\cite{chainspace}.)
\end{theorem}
\begin{proof}
We rely on the liveness properties of the byzantine agreement (shards with only $f$ nodes eventually reach consensus on a sequence), and the broadcast from nodes of shards to all other nodes of shards, channelled through the Transaction Manager. Assuming $T$ has been given to an honest node, it will be sequenced withing an honest shard BFT sequence, and thus a \preacceptts or \preabortts will be sent from the $2f+1$ honest nodes of this shard, aggregated into \acceptts or \abortts, and sent to the $f+1$ nodes of the other concerned shards. Upon receiving these messages the honest nodes from other shards will process the transaction within their shards, and the BFT will eventually sequence it. Thus the user will eventually receive a decision from at least $f+1$ nodes of a shard.
\end{proof}
\begin{theorem}
(Consistency~\cite{chainspace}) Under Honest Shards assumption, no two conflicting transactions, namely transactions sharing the same input will be committed. Furthermore, a sequential executions for all transactions exists. (From \chainspace~\cite{chainspace}.)
\end{theorem}
\begin{proof}
A transaction is accepted only if some nodes receive \acceptts, which presupposes all shards have provided enough evidence to conclude \preacceptts for each of them. Two conflicting transaction, sharing an input, must share a shard of at least $3f+1$ concerned nodes for the common object---with at most $f$ of them being malicious. Without loss of generality upon receiving the \preacceptts message for the first transaction, this shard will sequence it, and the honest nodes will emit messages for all---and will lock this object until the two phase protocol concludes. Any subsequent attempt to \preacceptts for a conflicting $T'$ will result in a \preabortts and cannot yield a accept, if all other shards are honest majority too. After completion of the first \acceptts the shard removes the object from the active set, and thus subsequent $T'$ would also lead to \preabortts. Thus there is no path in the chain of possible interleavings of the executions of two conflicting transactions that leads to them both being committed.
\end{proof}

\begin{theorem}
(Validity~\cite{chainspace}) Under Honest Shards assumption, a transaction may only be accepted if it is valid according to the smart contract (or application) logic. (From \chainspace~\cite{chainspace}.)
\end{theorem}
\begin{proof}
A transaction is committed only if some nodes conclude that \acceptts, which presupposes all shards have provided enough evidence to conclude \preacceptts for each of them. The concerned nodes include at least one shard per input object for the transaction; for any contract logic represented in the transaction, at least one of those shards will be managing object from that contract. Each shard checks the validity rules for the objects they manage (ensuring they are active) and the contracts those objects are part of (ensuring the transaction is valid with respect to the contract logic) in order to \preacceptts. Thus if all shards say \preacceptts to conclude that \acceptts, all object have been checked as active, and all the contract calls within the transaction have been checked by at least one shard---whose decision is honest due to at most $f$ faulty nodes. If even a single object is inactive or locked, or a single trace for a contract fails to check, then the honest nodes in the shard will emit \preabortts, and the final decision will be \abortts.
\end{proof}

\end{document}